\documentclass{article}
\usepackage{amsthm}
\usepackage{amsfonts}
\usepackage{epsfig}
\usepackage[dvips]{color}
\usepackage[usenames,dvipsnames]{xcolor}
\def\GG {{\cal G}}

\def\SS {{\cal S}}
\def\NN {{\mathbb N}}
\def\RR {{\mathbb R}}

\newtheorem{theorem}{Theorem}
\newtheorem{lemma}[theorem]{Lemma}
\newtheorem{proposition}[theorem]{Proposition}
\newtheorem{corollary}[theorem]{Corollary}
\newcount\claimno
\def\newclaim#1#2{
   \global\advance\claimno by 1\relax
   \bigskip\noindent\rlap{\rm(\the\claimno)}\ignorespaces
   \global\expandafter\edef\csname CLAIMLABEL#1\endcsname{\the\claimno}\relax
   \hangindent=33pt\hskip30pt{\sl#2}\bigskip\relax}
\def\refclaim#1{\csname CLAIMLABEL#1\endcsname}

\def\rt#1{{\color{red}#1}}
\def\rt#1{#1}
\def\rtt#1{{\color{green}#1}}
\def\rtt#1{#1}
\def\zz#1{{\color{red}#1}}
\def\zz#1{#1}
\def\junk#1{}

\begin{document}
\date{}
\newcounter{lth}
\setcounter{lth}{2}
\title{Testing first-order properties for subclasses 
of sparse graphs%
\thanks{A preliminary version of this paper appeared in FOCS 2010.}}
\author{Zden\v{e}k Dvo\v r\'ak\thanks{Computer Science Institute, Faculty of Mathematics and Physics, Charles University, Prague, Czech Republic. E-mail: {\tt rakdver@iuuk.mff.cuni.cz}. This author was supported by project 1M0545 of Ministry of~Education of~Czech Republic.}
\and
Daniel Kr\'al'\thanks{Mathematics Institute, DIMAP and Department of Computer Science, University of Warwick, Coventry CV4 7AL. E-mail: {\tt D.Kral@warwick.ac.uk}. Previous affiliation: Computer Science Institute, Faculty of Mathematics and Physics, Charles University, Prague, Czech Republic. Part of the work leading to this invention has received funding from the European Research Council under the European Union's Seventh Framework Programme (FP7/2007-2013)/ERC grant agreement no.~259385.}
\and
Robin Thomas\thanks{School of Mathematics, Georgia Institute of Technology, 
Atlanta, GA. Partially supported by NSF under No.~DMS-0739366.}
}

\maketitle

\begin{abstract}
We present a linear-time algorithm for deciding first-order (FO)
properties in classes of graphs with bounded expansion,
a notion recently introduced by Ne\v{s}et\v{r}il and Ossona de Mendez.
\rt{This generalizes several results from the literature, because}
many natural classes of graphs have bounded expansion:
graphs of bounded tree-width,
all proper minor-closed classes of graphs, graphs of bounded degree,
graphs with no subgraph isomorphic to a subdivision of a fixed graph,
and graphs that can be drawn in a fixed surface in such a way that
each edge crosses at most a constant number of other edges.
We deduce that there is an almost linear-time algorithm for deciding
FO properties in classes of graphs with locally bounded expansion.

More generally, we design a dynamic data structure for graphs belonging
to a fixed class of graphs of bounded expansion. 
After a linear-time initialization the data structure 
allows us to test an FO property in constant time, and
the data structure can be updated in constant time after
addition/deletion of an edge, provided the list of possible edges
to be added is known in advance and
their simultaneous addition results in a graph in the class.
All our results also hold for relational structures and
are based on the seminal result of Ne{\v s}et{\v r}il and
Ossona de Mendez on the existence of low tree-depth colorings.
\end{abstract}

\section{Introduction}

A celebrated theorem of Courcelle~\cite{courcelle} states that for every
integer $k\ge1$ and every property $\Pi$ definable in monadic
second-order logic (MSOL) there is a linear-time algorithm to
decide whether a graph of tree-width at most $k$ satisfies $\Pi$.
While the theorem itself is probably not useful in practice because
of the large constants involved, it does provide an easily verifiable
condition that a certain problem is (in theory) efficiently solvable.
Courcelle's result led to the development of a whole new area of algorithmics,
known as algorithmic meta-theorems; 
see the \rt{surveys~\cite{grokremeth,kreutzer-survey}}.
For specific problems there is very often a more efficient implementation,
for instance following the axiomatic approach of~\cite{GraphMinors13}.

While the class of graphs of tree-width at most $k$ is fairly large,
it does not include some important graph classes, such as planar graphs
or graphs of bounded degree.
Courcelle's theorem cannot be extended to these classes unless P=NP,
because testing $3$-colorability is NP-hard for planar
graphs of maximum degree at most four~\cite{gjs}
and yet $3$-colorability is expressible in monadic second order logic.

Thus in an attempt at enlarging the class of input graphs,
we have to restrict the set of properties  to be tested.
One of the first results in this direction was a linear-time algorithm of
Eppstein~\cite{bib-eppstein95,bib-eppstein99} for testing the existence of
a fixed subgraph in planar graphs. He then extended his algorithm
to minor-closed classes of graphs with locally bounded 
tree-width~\cite{bib-eppstein00}.
Since testing containment of a fixed subgraph can be expressed in first 
order logic by a $\Sigma_1$-sentence, this can be regarded
as a precursor to first order (FO) property testing.
Prior to our work, the following were the 
most general results:
\begin{itemize}
\item a linear-time algorithm of Seese~\cite{Seelinear} to test FO
      properties of graphs of bounded degree,
\item a linear-time algorithm of Frick and Grohe~\cite{fg} for deciding 
      FO properties of planar graphs,
\item an almost linear-time algorithm of Frick and Grohe~\cite{fg} 
      for deciding FO properties for classes of graphs with locally bounded tree-width,
\item a fixed parameter algorithm of Dawar, Grohe and Kreutzer~\cite{dawar} 
      for deciding FO properties for classes of graphs locally excluding a minor, and
\item a linear-time algorithm of Ne{\v s}et{\v r}il and 
Ossona de Mendez~\cite{bib-nes2}
      for deciding $\Sigma_1$-properties for classes of graphs with bounded expansion.
\end{itemize}
Our main theorem \rt{and its corollary generalize}  these five results.
In order to state \rt{them} we need a couple of definitions.
\rt{All graphs and digraphs in this paper are finite and have no loops 
or parallel edges. However, digraphs are permitted to have two edges joining
the same pair of vertices in opposite directions.}
For an integer $r\ge0$, a graph $H$ is an {\em $r$-shallow minor} of 
a graph $G$
if $H$  can be obtained from a subgraph of $G$ by 
contracting vertex-disjoint subgraphs of radii at most $r$
(and removing the resulting loops and parallel edges).
A class $\GG$ of graphs has {\em bounded expansion}
if there exists a function $f:\NN\to \RR^+$ such that for every integer
$r\ge0$ every $r$-shallow minor \rtt{$G$} of a member of $\GG$ 
\rtt{satisfies $|E(G)|/|V(G)|\le f(r)$.}
A preliminary version of our main theorem can be stated as follows.

\begin{theorem}
\label{thm-bounded}
Let $\GG$ be a class of graphs with bounded expansion,
and let $\Pi$ be a first-order property of graphs.
Then there exists a linear-time
algorithm that decides whether a graph from $\GG$
satisfies $\Pi$.
\end{theorem}

In fact, we prove a more general theorem \rt{(Theorem~\ref{thm-bounded2})}:
there exists a linear-time
algorithm for $L$-structures ``guarded" by a member of $\GG$,
and we design  several data structures that allow the $L$-structure
to be modified and support FO property testing in constant time.

Using known techniques we derive the following corollary from 
Theorem~\ref{thm-bounded}.
A class $\GG$ of graphs has {\em locally bounded expansion}
if there exists a function $g:\NN\times\NN\to\RR^+$ such that 
for every two integers $d,r\ge0$, for every graph $G\in\GG$ and for every
$v\in V(G)$, every $r$-shallow minor \rtt{$H$} of the $d$-neighborhood of $v$ in $G$
\rtt{satisfies $|E(H)|/|V(H)|\le g(d,r)$.}
We say that there exists 
an {\em almost linear-time algorithm} to solve a problem $\Pi$
if for every $\varepsilon>0$ there exists an algorithm to solve $\Pi$
with running time $O(n^{1+\varepsilon})$,
where $n$ is the size of the input instance.

\begin{corollary}
\label{thm-nowhere}
Let $\GG$ be a class of graphs with locally bounded expansion,
and let $\Pi$ be a first-order property of graphs.
Then there exists an almost linear-time algorithm that correctly decides
whether  a graph from $\GG$ satisfies $\Pi$.
\end{corollary}

We announced our results in the survey paper~\cite{wgsurvey}.
Dawar and Kreutzer~\cite{dak} posted an independent proof of 
Theorem~\ref{thm-bounded}
and a proof of Corollary~\ref{thm-nowhere} for the more general classes of
nowhere-dense graphs (introduced below).
However, the proofs in~\cite{dak} are incorrect.
A correct proof of 
Theorem~\ref{thm-bounded}, different from ours,
appears in~\cite{grokremeth}.

Thus it remains an interesting open problem whether Corollary~\ref{thm-nowhere}
can be generalized to the more general classes of nowhere-dense graphs.
This is of substantial interest from the point of view of fixed parameter
tractability, because nowhere density  of classes of
graphs gives a natural limitation (subject to a widely believed
complexity-theory assumption).
Indeed, we prove the following in Theorem~\ref{thm-complexity} below.
Let $L$ be a language consisting of one binary relation symbol and
let $\GG$ be a class of graphs closed under taking subgraphs that is
not nowhere dense.
We prove that if testing whether an input graph from $\GG$ satisfies
a given $\Sigma_1$-sentence $\varphi$ is fixed parameter tractable
when parameterized by the size of $\varphi$, then  FPT=W[1].

In the rest of this section we introduce  terminology and
state all our results.

\subsection{Logic theory definitions}
\label{sub-logic}

Our logic terminology is standard, except for the following.
All function symbols have arity one, and hence all functions are
functions of one variable.
If $L$ is a language, then
an $L$-term is {\em simple} if it is a variable or it is of the form $f(x)$
where $f$ is a function symbol and $x$ is a variable.
An $L$-formula is {\em simple} if all terms appearing in it are simple.
The rest of our logic terminology is standard, and so readers familiar
with it may skip the rest of this subsection.

A {\em language $L$} consists of a disjoint union of a finite set $L^r$ of {\em relation symbols} and
a finite set $L^f$ of {\em function symbols}. Each relation symbol $R\in L^r$ is associated
with an integer $a(R)\ge 0$, called the arity of $R$. 
In this paper all function symbols have arity one.

If $L$ is a language, then an {\em $L$-structure $A$} is a triple $(V,(R^A)_{R\in L^r},(f^A)_{f\in L^f})$
consisting of a finite set $V$ and
for each $m$-ary relation symbol $R\in L^r$ a set $R^A\subseteq V^m$,
the {\em interpretation} of $R$ in $A$, and for each function symbol $f\in L^f$
a function $f^A:V\to V$ of one variable, the {\em interpretation} of $f$ in $A$.
We define $V(A):=V$.
For example, graphs may be regarded as $L$-structures,
where $L$ is the language consisting of a single binary relation.
We define the size $|A|$ of $A$ to be $|V(A)|+\sum_{R\in L^r} |R^A|+|L^f||V(A)|$.
If $L$ contains no function symbols, then
an {\em $L$-substructure} of an $L$-structure $A=(V,(R^A)_{R\in L^r})$
is an $L$-structure $A'=(V',(R^{A'})_{R\in L^r})$ where $V'\subseteq V$ and
$R^{A'}\subseteq R^A\cap V'^{a(R)}$.
A language $L'$ extends a language $L$ if every function symbol of $L$
is a function symbol of $L'$ and the same holds for relation symbols, which also retain the same arity.
If a language $L'$ extends a language $L$, $A$ is an $L$-structure
and $A'$ is an $L'$-structure such that $V(A)=V(A')$ and $A$ and $A'$
have the same interpretations of symbols of $L$, then we say that
$A'$ is an {\em expansion} of $A$.

Assume that we have an infinite set of variables.
An {\em $L$-term} is defined as follows:
\begin{enumerate}
\item each variable is an $L$-term, and
\item if $f\in L^f$ and $t$ is an $L$-term, then $f(t)$ is an $L$-term.
\end{enumerate}
Each $L$-term is obtained by a finite number of applications of these two rules.
We say that an $L$-term is {\em simple} if it is a variable or is of the from $f(x)$
where $f\in L^f$ and $x$ is a variable. A term $t$ {\em appears} in a term $t'$
if either $t=t'$ or $t'=f(t'')$ for some $f\in L^f$ and $t$ appears in $t''$.

An {\em atomic $L$-formula} $\varphi$ is either the symbol 
$\top$ (which represents a tautology); or its negation $\bot$;
or $R(t_1,\ldots,t_m)$, where $R$ is an $m$-ary relation symbol of 
$L$ and $t_1,\ldots,t_m$ are $L$-terms; or
$t_1=t_2$, where $t_1$ and $t_2$ are $L$-terms. 
A term $t$ appears in $\varphi$ if it appears
in one of the terms $t_1,\ldots,t_m$.
An {\em $L$-formula} is defined recursively as follows: every atomic $L$-formula is an $L$-formula, and
if $\varphi_1$ and $\varphi_2$ are $L$-formulas and $x$ is a variable, then $\neg\varphi_1$, $\varphi_1\lor\varphi_2$,
$\varphi_1\land\varphi_2$, $\exists x\;\varphi_1$ and $\forall x\;\varphi_1$ are $L$-formulas.
Every $L$-formula is obtained by a finite application of these rules.
We write $t_1\neq t_2$ as a shortcut for $\neg(t_1=t_2)$.

A term $t$ {\em appears} in an $L$-formula $\varphi_1\lor\varphi_2$
if it appears in $\varphi_1$ or $\varphi_2$, and we define appearance for the other cases
analogously. An $L$-formula is {\em simple} if all terms appearing in it are simple.
A variable $x$ {\em appears freely} in an $L$-formula $\varphi$
if either $\varphi$ is atomic and $x$ appears in $\varphi$; or
$\varphi=\varphi_1\lor\varphi_2$ or $\varphi=\varphi_1\land\varphi_2$ and
$x$ appears freely in at least one of the formulas $\varphi_1$ and $\varphi_2$;
or $\varphi=\exists y\;\varphi'$ or $\varphi=\forall y\;\varphi'$,
$x$ is distinct from $y$ and $x$ appears freely in $\varphi'$.
Occurrences of $x$ in the formula $\varphi$ not inside the scope of a quantifier bounding $x$,
i.e., those that witness that $x$ appears freely in $\varphi$, are called {\em free} and
the variables that appear freely in $\varphi$ are also referred to as {\em free variables}.
If $\varphi$ is a formula, then the notation $\varphi(x_1,\ldots,x_n)$ indicates that
all variables that appear freely in $\varphi$ are among $x_1,\ldots,x_n$.
An {\em $L$-sentence} is an $L$-formula such that no variable appears freely in it.
A {\em $\Sigma_1$-$L$-sentence} is an $L$-formula of the form $\exists x_1,\ldots,x_n \varphi(x_1,\ldots,x_n)$
where the $L$-formula $\varphi(x_1,\ldots,x_n)$ is quantifier-free.
If $\varphi(x_1,\ldots,x_n)$ is an $L$-formula and $A$ is an $L$-structure,
then for $v_1,\ldots,v_n\in V(A)$, we define $A\models\varphi(v_1,\ldots,v_n)$ in the usual way.
We denote the {\em length} of a formula $\varphi$ by $|\varphi|$.
\rt{Finally, a property $\Pi$ of $L$-structures is called 
a {\em first order property} if there exists an $L$-sentence $\varphi$ 
such that every $L$-structure $A$ 
has property $\Pi$ if and only if $A\models\varphi$.
The property $\Pi$ is a {\em $\Sigma_1$-property} if $\varphi$ 
can be chosen to be a $\Sigma_1$-$L$-sentence.}

\subsection{Classes of sparse graphs}

The notion of a class of graphs of bounded expansion was
introduced by Ne{\v s}et{\v r}il and Ossona de Mendez in~\cite{bib-nes0} and
in the series of journal papers~\cite{bib-nes1,bib-nes2,bib-nes3}.
Examples of classes of graphs
with bounded expansion include proper minor-closed classes of graphs,
classes of graphs with bounded maximum degree, classes of graphs excluding
a subdivision of a fixed graph, classes of graphs that can be embedded
on a fixed surface with bounded number of crossings per edge and
many others, see~\cite{osmenwood}. Many structural and algorithmic properties
generalize from proper minor-closed classes of graphs to classes
of graphs with bounded expansion, see~\cite{wgsurvey,nessurvey}.

\rtt{Ne\v{s}et\v{r}il and Ossona de Mendez~\cite{bib-nowhere} defined}
a class $\GG$ of graphs \rtt{to be} {\em nowhere-dense} if
\rtt{%
for every $\varepsilon>0$ and every integer $r$ there exists a real number $K$
such that if $G$ is an $r$-shallow minor of a member of $\GG$, then
$|E(G)|\le K|V(G)|^{1+\varepsilon}$.
(It follows from~\cite[Corollary~3.3]{bib-nowhere} that this definition is
indeed equivalent to the one given in~\cite{bib-nowhere}.)}
It can be shown that every class of graphs with (locally) bounded expansion is 
nowhere-dense~\cite{bib-nowhere},
but the converse is false: \rt{the} class $\GG$ of graphs $G$ with no cycles of length
less than $\Delta(G)$ is nowhere-dense but it fails to have bounded expansion;
the class of graphs obtained from graphs $G$ in $\GG$ by adding a vertex
adjacent to all vertices of $G$ is a class of nowhere-dense graphs that does
not have locally bounded expansion.
One can also define a ``locally nowhere-dense" class of graphs, but it
turns out that such classes are nowhere-dense~\cite{bib-nowhere}.

If $L$ is a language, then
the {\em Gaifman graph} of an $L$-structure $A$ is the undirected graph
$G_A$ with vertex set $V(G_A)=V(A)$ and an edge between two \rt{distinct} vertices 
$a, b\in V(A)$ 
if and only if there exist $R\in L^r$ and a tuple 
$(a_1, \ldots , a_r)\in R^A$ such that $a, b\in \{a_1, \ldots , a_r\}$ or
there exists a function $f\in L^f$ such that $b=f^\rtt{A}(a)$ or $a=f^\rtt{A}(b)$.
We say that the relational structure $A$ is {\em guarded} by a graph $G$
if $V(G)=V(A)$ and $G_A$ is a subgraph of $G$.
Observe that if $G$ belongs to a class of graphs with bounded expansion,
then every subgraph of $G$ has a vertex of bounded degree, and hence 
the number  of complete subgraphs of $G$ is linear in $|V(G)|$
by a result of~\cite{woodcl}. It follows that
the size $|A|$ of any $L$-structure $A$ guarded by a graph belonging
to a fixed class of graphs with bounded expansion is $O(|V(A)|)$.

\rtt{%
Our model of computation is the standard RAM model with addition
and subtraction as arithmetic operations.
An $L$-structure $A$ is represented in the straightforward way by listing
all elements of $V(A)$, the images of elements under functions of $A$,
and listing all tuples of all relations of $A$.
However, we will need to be able to decide in constant time whether
a given tuple satisfies the interpretation of a relation in an 
$L$-structure, and we now explain how to do that for $L$-structures
guarded by $d$-degenerate graphs.
Let $d$ be a fixed integer. A graph $G$ is called {\em $d$-degenerate}
if every subgraph of $G$ has a vertex of degree $d$ or less.
Thus if $\GG$ is a class of graphs of bounded expansion, then there
exists an integer $d$ such that every member of $\GG$ is $d$-degenerate.
Now let $A$ be an $L$-structure guarded by a $d$-degenerate graph $G$.
Since $G$ is $d$-degenerate, its vertices can be numbered
$v_1,v_2,\ldots,v_n$ in such a way that for each $i=1,2,\ldots,n$
the vertex $v_i$ has at most $d$ neighbors among $v_{1},\ldots,v_{i-1}$.
Now each $t$-tuple $(v_{i_1},v_{i_2},\ldots,v_{i_t})$, where 
$i_1<i_2<\cdots<i_t$, will be associated with the vertex $v_{i_t}$.
Then each vertex is associated with at most $d\choose t-1$ distinct
$t$-tuples. For each relation we compute the associations at the
beginning of the computation, and then we can answer in constant time
queries of the form whether a given $t$-tuple  belongs to a given relation.}

\subsection{Our results}

We first state versions of Theorem~\ref{thm-bounded}
and Corollary~\ref{thm-nowhere} for $L$-structures.
\rt{Theorem~\ref{thm-bounded} and Corollary~\ref{thm-nowhere}
are immediate consequences.}

\begin{theorem}
\label{thm-bounded2}
Let $\GG$ be a class of graphs with bounded expansion,
$L$ a language and
$\varphi$ an $L$-sentence. There exists a linear-time
algorithm that decides whether an $L$-structure guarded by a graph $G\in\GG$
satisfies $\varphi$.
\end{theorem}

\begin{corollary}
\label{thm-nowhere2}
Let $\GG$ be a class of graphs with locally bounded expansion,
$L$ a language and $\varphi$ an $L$-sentence.
There exists an almost linear-time algorithm that decides
whether an $L$-structure guarded by a graph $G\in\GG$ satisfies $\varphi$.
\end{corollary}

Our approach differs from the methods used to prove the results 
from~\cite{dawar,fg,Seelinear} mentioned above and is based on 
a seminal result of Ne{\v s}et{\v r}il and
Ossona de Mendez~\cite{bib-nes1} on the existence of low tree-depth colorings
for graphs with bounded expansion,
stated below \rt{in a form suitable for our purposes}
as Theorem~\rtt{\ref{lem:lowtdepth}}.

We also consider dynamic setting and
design the following data structures, where the last one can be viewed as
a dynamic version of Theorem~\ref{thm-bounded2}.
\begin{itemize}
\item For every class $\GG$ of graphs with bounded expansion, an integer $d_0$ and
      a language $L$, \rt{we design}
      a data structure such that \rt{given a graph $G\in\GG$ on $n$ vertices
      and an $L$-structure guarded by $A$} the data structure is initialized 
      in time $O(n)$ and supports:
      \begin{itemize}
      \item adding a tuple to a relation of $A$ in time $O(1)$ provided $A$ stays guarded by $G$,
      \item removing a tuple from a relation of $A$ in time $O(1)$,
      \item answering whether $A\models\varphi$ for \rtt{a} $\Sigma_1$-$L$-sentence $\varphi$
            with at most $d_0$ variables in time $O(|\varphi|)$ and
            outputting one of the satisfying assignments, and
      \end{itemize}
\rt{%
\item for every class $\GG$}
      of nowhere-dense graphs, 
      \rt{every integer $d_0$, and every language $L$}
      \rt{we design a}  data structure
      \rt{such that 
      for every $\varepsilon>0$,
      given a graph $G\in\GG$ on $n$ vertices
      and an $L$-structure guarded by $A$, the data structure}
      is initialized in time $O(n^{1+\epsilon})$ and supports:
      \begin{itemize}
      \item adding a tuple to a relation of $A$ in time $O(n^{\varepsilon})$ provided $A$ stays guarded by $G$,
      \item removing a tuple from a relation of $A$ in time $O(n^{\varepsilon})$, and
      \item answering whether $A\models\varphi$ for \rtt{a} $\Sigma_1$-$L$-sentence $\varphi$
            with at most $d_0$ variables in time 
            $O(|\varphi|)$ and
            if so, outputting one of the satisfying assignments, and
      \end{itemize}
\item for every class $\GG$ of graphs with bounded expansion,
      a language $L$ and an $L$-sentence $\varphi$,
      \rt{we design}
      a data structure such that 
      \rt{that given a graph $G\in\GG$ on $n$ vertices
      and an $L$-structure guarded by $A$ the data structure}
      is initialized 
      in time $O(n)$ and supports:
      \begin{itemize}
      \item adding a tuple to a relation of $A$ in time $O(1)$ provided $A$ stays guarded by $G$,
      \item removing a tuple from a relation of $A$ in time $O(1)$,
      \item answering whether $A\models\varphi$ in time $O(1)$.
      \end{itemize}
\end{itemize}
The first of these data structures is needed in our linear-time algorithm
for $3$-coloring triangle-free graphs on surfaces~\cite{ourcol}, 
also see~\cite{oursoda}.
\rt{%
The first two data structures are presented in Theorems~\ref{thm-databe}
and~\ref{thm-databe-nowhere} in Section~\ref{sec-subgraph},
and the third one is presented in Theorem~\ref{thm-dynamic+} in
Section~\ref{sec-dynamic}.
}

\subsection{A hardness result}

Theorem~\ref{thm-bounded} and Corollary~\ref{thm-nowhere} fall within
the realm of fixed parameter tractability (FPT). We say that a decision
problem $\Pi$ parameterized by a parameter $t$ is {\em fixed parameter tractable}
if there exists an algorithm for $\Pi$ with running time $O(f(t)n^c)$,
where $n$ is the size of the input, $f$ is an arbitrary function and
$c$ is a constant independent of $t$. Analogously to the polynomial
hierarchy starting with the classes P and NP, there exists
a hierarchy of classes $\mbox{FPT}\subseteq\mbox{W[1]}\subseteq\mbox{W[2]}\subseteq\cdots$ of
parameterized problems, 
where FPT is the class of problems that are fixed parameter tractable.
See e.g.~\cite{bib-downey+fellows,bib-flum+grohe,bib-niedermeier} for more details.

\begin{theorem}
\label{thm-complexity}
Let $\GG$ be a class of graphs closed under taking subgraphs.
If $\GG$ is not nowhere-dense
and the problem of deciding $\Sigma_1$-properties in $\GG$
is fixed parameter tractable \rtt{when parametrized by the length
of the formula that defines the property}, then {\rm FPT=W[1]}.
\end{theorem}

\begin{proof}
An $r$-subdivision of a graph $G$ is the graph obtained from $G$
by subdividing every edge exactly $r$ times.
Since $\GG$ is not nowhere-dense and $\GG$ is closed under taking subgraphs,
there exists an integer $r$ such that
$\GG$ contains an $r$-subdivision of every graph~\cite{bib-nowhere}.
Since the existence of a subgraph isomorphic to an $r$-subdivision of
the complete graph $K_m$ is a $\Sigma_1$-property for every $m$,
we derive from the hypothesis of the theorem that
there exists an FPT algorithm $\cal A$ to decide the existence 
of an $r$-subdivision of
the complete graph $K_m$ in an input graph from $\GG$,
where \rt{the problem is} parameterized by $m$. 
This implies that testing the existence of
a complete subgraph of order $m$ is fixed parameter tractable
for general graphs $G$, because it is equivalent to testing whether
the $r$-subdivision of $G$ has a subgraph isomorphic to the 
$r$-subdivision of $K_m$, and the latter can be tested using the algorithm $\cal A$.
But testing the existence of a $K_m$ subgraph
is a well-known W[1]-complete problem~\cite{bib-w1}, and hence
FPT=W[1], as desired.
\end{proof}

Dawar and Kreutzer~\cite{dak} proved the related result
that if $\GG$ fails to be nowhere dense in an ``effective" way
and deciding FO properties in $\GG$
is fixed parameter tractable,
then FPT=AW[$*$]. 

The paper is organized as follows.
In the next section we review results about classes 
of graphs with bounded expansion
and classes of nowhere-dense graphs that will be needed later.
In Section~\ref{sec-formulas} we prove Theorem~\ref{thm-bounded2}
and in Section~\ref{sec-locally} we use it to deduce 
Corollary~\ref{thm-nowhere2}.
In Section~\ref{sec-subgraph} we present the first two data structures
mentioned earlier in this section, and in the final
Section~\ref{sec-dynamic} we present the third data structure.

\zz{A conference version of this article appeared in~\cite{DvoKraThoFOCS}.}

\section{Classes of graphs with bounded expansion}

In this section, we survey results on classes of graphs with bounded expansion 
and classes of nowhere-dense graphs that we need in the paper.
Let $G$ be a graph, and let $r\ge0$ be an integer.
Let us recall that a graph $H$ is an $r$-shallow minor of $G$ if
$H$ can be obtained from a subgraph of $G$ by contracting vertex-disjoint
subgraphs of radii at most $r$ and deleting the resulting loops and
parallel edges.
Following Ne\v{s}et\v{r}il and Ossona de Mendez we denote by
$\nabla_r(G)$ the maximum of \rtt{$|E(G')|/|V(G')|$ over all}  
$r$-shallow minors \rtt{$G'$} of $G$.
Thus $\nabla_0(G)$ is the 
maximum of $|E(G')|/|V(G')|$ taken over all subgraphs $G'$ of $G$.
Since every subgraph of  $G$ 
has a vertex of degree at most $2\nabla_0(G)$, 
we see that $G$ is $(2\nabla_0(G)+1)$-colorable and that
it has an orientation with maximum in-degree at most $2\nabla_0(G)$.
Clearly, such an orientation can be found in linear time in a greedy way.
Thus we have the following.

\begin{lemma}
\label{lem:bound}
For every class $\GG$ of graphs of bounded expansion there exists
an integer $K$ such that every graph $G\in\GG$ is $K$-colorable and
has an orientation with maximum in-degree at most $K-1$.
Furthermore, a $K$-coloring of $G$ and an orientation with maximum
in-degree $K-1$ can be found in linear time.
\end{lemma}

Let $D$ be a directed graph, and let $D'$ be a directed graph obtained
from $D$ by adding, for every pair of vertices $x,y\in V(D)$,
\begin{itemize}
\item the edge $xy$ if
      $D$ has no edge from $x$ to $y$ and
      there exists a vertex $z\in V(D)$ such that $D$ has
      an edge oriented from $x$ to $z$ and an edge oriented from $z$ to $y$
      ({\em transitivity}), and
\item either the edge $xy$ or the edge $yx$ 
      if $x$ is not adjacent to $y$ and
      there exists a vertex $z$ such that $D$ has
      an edge oriented from $x$ to $z$ and an edge oriented from $y$ to $z$
      ({\em fraternality}).
\end{itemize}
We call $D'$ an {\em oriented augmentation} of $D$ and
the underlying undirected graph of $D'$ the {\em augmentation} of $D$.
The following is
a result of
Ne{\v s}et{\v r}il and Ossona de Mendez~\cite[Lemma~5.2]{bib-nes1}.
A self-contained proof may be found
in~\cite{wgsurvey}.

\begin{theorem}
\label{thm:augment}
There exist polynomials $f_0, f_1, f_2, \ldots$ with the following property.
Let $D$ be an orientation of an undirected  graph $G$,
let $D$ have maximum in-degree at most $\Delta$, 
and let $G'$ be the augmentation of $D$.
Then $\nabla_r(G')\le f_r(\nabla_{2r+1}(G),\Delta)$ for all $r\ge0$.
\end{theorem}

Let $G$ be a graph.
Consider the following sequence of directed graphs:
Let $D_0$ be an orientation of $G$ with maximum in-degree at most
$2\nabla_0(G)+2$ and assume that we have constructed $D_0,D_1,\ldots,D_{k-1}$.
(An alert reader may be wondering why we added the extra factor of $+2$ to
the bound on the in-degree. The reason for that will become clear
in the proof of Lemma~\ref{lm-simple}.)
Let $G_k$ be the augmentation of $D_{k-1}$, 
and let $D_k$ be an oriented augmentation of $D_{k-1}$ chosen in such
a way that the maximum in-degree of the subgraph formed by the edges
added according to the fraternality rule is at most $2\nabla_0(G_k)$.
This is possible, because $G_k$ itself has an orientation with maximum
in-degree at most $2\nabla_0(G_k)$.
We say that $G_k$ is a {\em $k$-th augmentation} of $G$.
If $D_{k-1}$ has in-degree at most $\Delta$, then $D_k$ has
in-degree at most $\Delta+\Delta^2+2\nabla_0(G_k)$,
and Theorem~\ref{thm:augment} implies that
$\nabla_r(G_k)\le f_r(\nabla_{2r+1}(G_{k-1}),\Delta)$ 
for all $r\ge0$.
Thus we arrive at the following result of Ne\v{s}et\v{r}il and
Ossona de Mendez.

\begin{theorem}
\label{thm:augmexp}
Let $\GG$ be a class of graphs, let $k\ge0$ be an integer, and let
$\GG_k$ be the class of all $k$-th augmentations of members of $\GG$.
\begin{itemize}
\item[(i)] If $\GG$ has bounded expansion, then $\GG_k$ has bounded
expansion.
\item[(ii)] If $\GG$ is nowhere dense, then $\GG_k$ is nowhere dense.
\end{itemize}
\end{theorem}

Please note that  statement (ii) above needs that the functions
$f_0,f_1,\ldots$ referenced in Theorem~\ref{thm:augment} are polynomials.
It also follows that if $G$ belongs to a class of graphs of bounded expansion,
then $D_1,D_2,\ldots,D_k$ and $G_k$ can be found in linear 
time. We state that as a lemma for future reference.

\begin{lemma}
\label{alg:augment}
For every class $\GG$ of graphs of bounded expansion and for every
fixed integer $k$ there exists a linear-time algorithm that
computes a $k$-th augmentation of $G\in\GG$ and the directed graphs
$D_1,D_2,\ldots,D_k$ as in the definition of $k$-th augmentation.
\end{lemma}

\rtt{For nowhere dense graphs we have  the following analogue.
\begin{lemma}
\label{alg:ndaugment}
For every nowhere dense class $\GG$ of graphs and for every
fixed integer $k$ there exists an almost linear-time algorithm that
computes a $k$-th augmentation of $G\in\GG$ and the directed graphs
$D_1,D_2,\ldots,D_k$ as in the definition of $k$-th augmentation.
\end{lemma}
}


An {\em out-branching} is a rooted \rt{directed} tree where every edge is directed
away from the root.
A {\em rooted forest} $F$ is a directed graph such that every
weak component is an out-branching (recall that a weak component
of a directed graph is any minimal subgraph with no incoming or
outgoing edge).
A {\em subforest} of $F$ is a subgraph $F'$ of $F$ such that
if $v$ is a vertex included \rtt{in} $F'$,
then the path between the root of the tree of $F$ containing $v$ and
the vertex $v$ is also contained in $F'$.
The {\em depth of a vertex} $v$ of a rooted forest $F$ is the number of vertices on
the path from the root of the tree containing $v$ \rt{to} the vertex $v$.
The {\em depth of a rooted forest} $F$ is the maximum depth of a vertex of $F$.
Finally, if $F$ is a rooted forest $F$,
then the {\em subtree} of a vertex $v$ is the subgraph of $F$ induced
by all vertices reachable from $v$.

The {\em closure} of a rooted forest $F$ is the undirected graph with vertex-set $V(F)$
and edge-set all pairs of \rt{distinct} vertices joined by a directed path in $F$.
The {\em tree-depth} of an (undirected) graph $G$
is the smallest integer $s$ such that $G$ is a subgraph of the closure of
a rooted forest of depth $s$.
For an integer $d\ge1$
a vertex coloring of a graph $G$ is a {\em low tree-depth
coloring of order $d$} if for every $s=1,2,\ldots,d$ the union of any $s$ color
classes induces a subgraph of $G$ of tree-depth at most $s$.
In particular, every low tree-depth coloring of order $d$ of $G$ is
a proper coloring of~$G$.
\rt{%
If $s\in\{1,2,\ldots,d\}$ and $H$ is the subgraph of $G$ induced by
$s$ color classes, then there exists a rooted forest $F$ of depth at most
$s$ such that $H$ is a subgraph of the closure of $F$.
If for all $s$ and all subgraphs $H$ as above the forest $F$ can be
chosen in such a way that its \rtt{closure} is a subgraph of
some fixed graph $G'$, then we say that the low tree-depth coloring is 
{\em $G'$-compliant},
and we refer to the corresponding forests $F$ as
{\em depth-certifying forests}.%
}


The following theorem follows from~\cite[Lemma~6.2]{bib-nes1}.
In the interest of clarity we give a proof.

\begin{theorem}
\label{lem:lowtdepth}
Let $G$ be a graph, let $d$ be an integer, let $k:=3(d+1)^2$,
let $G'$ be a $k$-th augmentation of $G$, and let $c$ be a proper coloring of $G'$.
Then $c$ is a $G'$-compliant low tree-depth coloring of $G$ of order $d$.
\end{theorem}

\proof
Let $s\in\{1,2,\ldots,d\}$ and let $H'$ be a subgraph of $G$ induced
by the union of $s$ color classes of the coloring $c$.
If $H'$ has tree-depth at least $s+1$, then let $H$ be a subgraph of $H'$
of tree-depth exactly $s+1$; otherwise, let $H:=H'$.
Since $H$ is $s$-colorable, it does not have a complete subgraph
on $s+1$ vertices.
By~\cite[Lemma~2]{wgsurvey} applied to every component of
$H$ and taking both $d$ and
$p$ in that lemma to be $s+1$ we deduce that $H$ has tree-depth at most $s$
and that it has a depth-certifying forest whose \rtt{closure}
is a subgraph of $G'$. Thus $H=H'$, and the lemma follows.~\qed
\medskip

In fact, it follows from the proof that the depth-certifying forests
can be found in linear-time, formally as follows.

\begin{theorem}
\label{alg:lowtdepth}
Let $d\ge1$ be an integer, and let $k:=3(d+1)^2$.
There exists an algorithm with the following specifications:\\
{\bf Input:} \rtt{An integer $s\in\{1,2,\ldots,d\}$,}
a graph $G$ and directed graphs $D_0,D_1,\ldots,D_k$ as in
the definition of $k$-th augmentation, 
\rtt{%
a proper coloring $c$ of the underlying undirected graph $G'$ of $D_k$ 
(so that $G'$ is a $k$-th augmentation of $G$), and a subgraph $H$ of
$G'$ that is the union of $s$ color classes of $c$.}\\
{\bf Output:} A rooted forest $F$ of depth at most $s$ such that
$H$ is a subgraph of the closure of $F$ and the \rtt{closure}
of $F$ is a subgraph of $G'$.\\
{\bf Running time:} $O(|V(G')|+|E(G')|)$.
\end{theorem}


\section{Deciding FO properties in linear time}\label{sec-formulas}

In this section, we prove Theorem~\ref{thm-bounded2}.
We start with a lemma which allows us to remove quantifiers
from an FO formula (Lemma~\ref{lm-elim}). However, we need more definitions.
Let $L$ be a language and let $X$ be a set of $L$-terms.
An {\em $X$-template $T$} is a rooted forest with vertex set $V(T)$
equipped with a mapping $\alpha_T:X\to V(T)$ such that
$\alpha_T^{-1}(w)\not=\emptyset$ for every vertex $w$ of $T$ with no descendants.
If $\varphi$ is a quantifier-free $L$-formula,
then a {\em $\varphi$-template} is an $X$-template
where $X$ is the set of all terms appearing in $\varphi$.
Two $X$-templates $T$ and $T'$ are isomorphic
if there exists a bijection $f:V(T)\to V(T')$ such that
\begin{itemize}
\item $f$ is an isomorphism of $T$ and $T'$ as rooted forests; in particular,
      $w$ is a root of $T$ if and only if $f(w)$ is a root of $T'$, and
\item $f(\alpha_T(t))=\alpha_{T'}(t)$ for every $L$-term $t\in X$.
\end{itemize}

The number of non-isomorphic $X$-templates of a given depth is finite, as
stated in the next proposition. The proof is straightforward and
is left to the reader.

\begin{proposition}
\label{prop-fin-template}
For every finite set of terms $X$ and every integer $d$,
there exists \rtt{an integer $K$ such that there are at most $K$}
non-isomorphic $X$-templates of depth at most $d$.
\end{proposition}

Let $L$ be a language and let $X$ be a set of $L$-terms with 
variables $\{x_1,\ldots, x_n\}$.
An {\em embedding} of an $X$-template $T$ in a rooted forest $F$
is a mapping $\nu:V(T)\to V(F)$ such that $\nu(r)$ is a root of $F$
for every root $r$ of $T$ and $\nu$ is an isomorphism of $T$ and
the subforest of $F$ with vertex set $\nu(V(T))$.
Let $S$ be an $L$-structure guarded by the closure of $F$, and $v_1,\ldots,v_n\in V(S)$.
We say that the embedding $\nu$ is {\em $(v_1,\ldots,v_n)$-admissible}
for $S$ if for every term $t(x_1,\ldots,x_n)\in X$, we have
$\nu(\alpha_T(t))=t(v_1,\ldots,v_n)$,
where $t(v_1,\ldots,v_n)$ denotes the element of $V(S)$
obtained by substituting $v_i$ for $x_i$ in the term $t$ and
evaluating the \rt{interpretations in $S$ of the} 
function \rt{symbols in}  the term $t$ 
(in particular, if $x_i\in X$, then $\nu(\alpha_T(x_i))=v_i$).
\rtt{%
We say that the elements $v_1,v_2,\ldots,v_n$ (in the order listed)
are \emph{compatible with $T,F$ and $S$} if there exists a
$(v_1,\ldots,v_n)$-admissible embedding of $T$ in $F$ for $S$.}
We will need the following lemma.

\begin{lemma}
\label{lm-extemp}
Let $d\ge1$ be an integer, let $F$ be a rooted forest of depth at most $d$,
let $L$ be a language, let $\varphi(x_1,\ldots,x_n)$ be a quantifier-free $L$-formula,
let $S$ be an $L$-structure guarded by the closure of $F$, and
let $v_1,\ldots,v_n\in V(S)$.
Then there exists a $\varphi$-template $T$ \rt{of depth at most $d$}
\rtt{such that $v_1,\ldots,v_n$ are compatible with $T,F$ and $S$.}
\end{lemma}

\begin{proof}
Let $X$ be the set of all $L$-terms that appear in $\varphi$, and let
$Y$ be the set of all evaluations $t(v_1,\ldots,v_n)$ of all terms
$t(x_1,\ldots,x_n)$ from $X$.
Let $T$ be the smallest subforest of $F$ that includes all vertices
from $Y$ and the root of every component of $F$ that includes an 
element of $Y$.
For a term $t(x_1,\ldots,x_n)$ in $X$ let $\alpha_T(t):=t(v_1,\ldots,v_n)$,
and let $\nu$ be the identity mapping \rt{$V(T)\to V(F)$}.
Then $T$ \rtt{is a $\varphi$-template} and $\nu$ 
\rtt{is a $(v_1,\ldots,v_n)$-admissible embedding of $T$ in $F$ for $S$,}
as desired.
\end{proof}

\rtt{We remark that in the previous lemma $T$ and $\nu$ are unique.}
If $F$ is a rooted forest, then \rtt{a} function $p:V(F)\to V(F)$
is the {\em $F$-parent function}
if $p(v)$ is the parent of $v$ unless $v$ is a root of $F$;
if $v$ is a root of $F$, $p(v)$ is set to be equal to $v$.

We now show that it can be tested by a quantifier-free 
formula whether \rtt{there exists an admissible embedding.}

\begin{lemma}[Testing admissibility]
\label{lm-embed}
Let $L$ be a language \rt{that includes} a function symbol $p$ and
let $X$ be a finite set of \rtt{$L$-}terms with variables $x_1, \ldots, x_n$.
If $T$ is an $X$-template, then
there exists a quantifier-free formula $\xi_T(x_1,\ldots,x_n)$ such that
for every rooted forest $F$ and every $L$-structure $S$ guarded by the closure of $F$
such that \rtt{the interpretation} $p^S$ \rtt{of $p$ in $S$}
is the $F$-parent function, and
for every $n$-tuple $v_1,\ldots,v_n\in V(S)$,
the $L$-structure $S$ satisfies $\xi_T(v_1,\ldots,v_n)$
if and only if \rtt{$v_1,\ldots,v_n$ are compatible with $T,F$ and $S$.}
\end{lemma}

\begin{proof}
Let $q:V(T)\to V(T)$ be the $T$-parent function,
\rtt{and let $d$ be the depth of $T$}.
Set $\xi_T(x_1,\ldots,x_n)$ to be the conjunction of all formulas
\begin{itemize}
\item $p^k(t)=p^{k'}(t')$ if $q^k(\alpha_T(t))=q^{k'}(\alpha_T(t'))$, and
\item $p^k(t)\not=p^{k'}(t')$ if $q^k(\alpha_T(t))\not=q^{k'}(\alpha_T(t'))$,
\end{itemize}
for all pairs of not necessarily distinct terms $t,t'\in X$ and
all pairs of integers $k$ and $k'$, $0\le k,k'\le d+1$ (note that
including the formulas with $t=t'$ allows for testing the depth
of $t$ in $F$).
Here $p^k$ denotes the function $p$ iterated $k$ times.

It is straightforward to show that for $v_1,\ldots,v_n\in V(S)$,
a $(v_1,\ldots,v_n)$-admissible embedding for $S$ of $T$ in $F$
exists if and only if $S\models\xi_T(v_1,\ldots,v_n)$.
\end{proof}

The following lemma is the core of our algorithmic arguments as it allows
replacing an existentially quantified subformula with a quantifier-free formula.
Recall that an $L$-term is simple if it is a variable or 
a function image of a variable, and  
an $L$-formula is simple if all terms appearing in it are simple.

\begin{lemma}
\label{lm-phi-T}
Let $d\ge 0$ be an integer, $L$ a language,
$\varphi(x_0,\ldots,x_n)$ a simple quantifier-free $L$-formula that
is a conjunction of atomic formulas and their negations, and
$T$ a $\varphi$-template.
There exist a language $\overline{L}$ that extends $L$ and
a \rt{(not necessarily simple)}
quantifier-free $\overline{L}$-formula $\overline{\varphi}_T(x_1,\ldots,x_n)$ 
such that the following holds:
\begin{itemize}
\item $\overline{L}$ is obtained from $L$ by adding a function symbol $p$
      and finitely many relation symbols $U_0,\ldots,U_K$ of arity 
      \zz{at most} one, and
\item for every rooted forest $F$ of depth at most $d$ and every $L$-structure $S$
      guarded by the closure of $F$,
      there exists an $\overline{L}$-structure $\overline{S}$
      such that $\overline{S}$ is an expansion of $S$,
      \rtt{$p^{\overline{S}}$ is the $F$-parent function and
      the relations $U^{\overline{S}}_0,\ldots,U^{\overline{S}}_K$
      can be computed (by listing the singletons they contain) in 
      linear time given $F$ and $S$}, and
       for all $v_1,\ldots,v_n\in V(S)$
      \begin{center}
      $\overline{S}\models\overline{\varphi}_T(v_1,\ldots,v_n)$
      if and only if
      $S\models\varphi(v_0,v_1,\ldots,v_n)$ for some $v_0\in V(S)$ such that
      $v_0,v_1,\ldots,v_n$ are compatible with $T,F$ and $S$.
      \end{center}
\end{itemize}
\end{lemma}

\begin{proof}
Let $T$ be a $\varphi$-template  of depth at most $d$,
let $q$ be the $T$-parent function, and 
let $X$ be the set of all terms appearing in $\varphi$.
Let $\xi_T$ be the formula from Lemma~\ref{lm-embed} applied to the language
obtained from $L$ by adding the function symbol $p$. 
\zz{%
We will have to distinguish two cases depending on whether the following
condition is satisfied:
}

{
\newclaim{stat0}{\zz{%
The tree of $T$ containing the vertex $\alpha_T(x_0)$
also contains an $\alpha_T$-image of a term in which another variable
appears.}
}

}

Let $K$ be an integer such that every vertex of $T$ has at most $K$
children and $T$ has at most $K$ weak components, and let
$\overline L$ be obtained from $L$ by adding a function symbol $p$
and relation symbols $U_0,\ldots,U_K$.
\zz{%
If (\refclaim{stat0}) holds, then $U_0,\ldots,U_K$ will have arity one;
otherwise $U_0$ will have arity one and $U_1,\ldots,U_K$ will have arity zero.}
The construction of $\overline{\varphi}_T(x_1,\ldots,x_n)$ will proceed
in several steps.

Let $t=f(x_i)$ be an \rt{$L$-}term appearing in $\varphi$, for some function symbol $f\in L^f$ and
a variable $x_i$ with $0\le i\le n$.
\rt{(Since $\varphi$ is simple, every $L$-term appearing in $\varphi$
is either a variable or of this form.)}
If $\alpha_T(t)$ is neither an ancestor nor a descendant of $\alpha_T(x_i)$, then
for every rooted forest $F$ of depth at most $d$, every $L$-structure $S$
guarded by the closure of $F$ and every
choice of $v_1,\ldots,v_n\in V(S)$,
there is no $(v_0,\ldots,v_n)$-admissible embedding for $S$ of $T$ into $F$,
because $v_i$ and $f^S(v_i)$ are adjacent in the Gaifman graph of $S$;
in particular, one is a descendant of the other in $F$.
Hence, \rt{if we} set $\overline{\varphi}_T$ to $\bot$,
\rt{then $\overline{\varphi}_T$ satisfies the conclusion of the lemma.
Since $\varphi$ is simple, we may}
assume the following:

{
\newclaim{stat1}{%
If the variable $x_i$ appears in a term $t\in X$, then $\alpha_T(t)$
is an ancestor or a descendant of $\alpha_T(x_i)$.
}

}


\rt{Assume now that} $\alpha_T(x_0)$ is an ancestor of a vertex $\alpha_T(t)$,
say $q^k(\alpha_T(t))=\alpha_T(x_0)$ for $k\ge 0$,
where $t\in X$ is an \rtt{$L$-}term such that $x_0$ does not appear in $t$.
\rt{In that case let} $\overline{\varphi}_T$
will be the formula obtained from $\varphi\land\xi_T$ by replacing
each $x_0$ with the term $p^k(t)$.
Clearly, \rt{for every $\overline L$-structure $\overline S$ that is
an expansion of $S$ we have}
$\overline{S}\models\overline{\varphi}_T(v_1,\ldots,v_n)$ if and only if 
there is a choice of $v_0$ in $V(F)$ such that
$S\models\varphi(v_0,\ldots,v_n)$ and  $v_0,\ldots,v_n$
\rtt{are compatible with $T,F$ and $S$.}
\rt{Since $\varphi$ is simple}, we \rt{may} assume the following:

{
\newclaim{stat2}{%
Every \rtt{$L$-}term $t\in X$ such that $\alpha_T(t)$ is contained in the subtree of 
$\alpha_T(x_0)$ is $x_0$ or a function image of $x_0$.
}

}


We now define an auxiliary formula $\varphi'$
to be the formula obtained from $\varphi$ by replacing all atomic formulas of the form:
\begin{itemize}
\item $t=t'$, where $t$ and $t'$ are terms such that $\alpha_T(t)\not=\alpha_T(t')$, and
\item $R(t_1,\ldots,t_m)$ such that $\alpha_T(t_1),\ldots,\alpha_T(t_m)$
      are not \rt{the} vertices of a clique in the closure of $T$,
\end{itemize}
by $\bot$. 

{
\newclaim{stat4}{%
Let $S$ be an $L$-structure guarded by the closure of a rooted forest
$F$, and let
there exist a $(v_0,\ldots,v_n)$-admissible embedding $\nu$ of $T$ in $F$ for $S$.
Then $S\models\varphi(v_0,\ldots,v_n)$ if and only if 
$S\models\varphi'(v_0,\ldots,v_n)$.
}

}

\rtt{%
\noindent
We notice that, by the existence of $\nu$,
the atomic formulas that got replaced by the definition of $\varphi'$ are
not satisfied by  $S$.
This proves (\refclaim{stat4}).
}
\bigskip

We will now complete the proof under the assumption 
\zz{that (\refclaim{stat0}) holds.}
Let $v$ be the nearest ancestor of $\alpha_T(x_0)$ in $T$
such that there exists a term $t_v\in X$ such that $x_0$ does not appear in $t_v$ and
$v$ is an ancestor of $\alpha_T(t_v)$. Note that $v\neq \alpha_T(x_0)$ by
(\refclaim{stat2}).
Let $d_v$ be the depth of $v$ in $T$, $d_{x_0}$ the depth of $\alpha_T(x_0)$ and
$m$ the number of children of $v$ in $T$.  Let $t_1,\ldots,t_{m-1}$ be terms such that
$\alpha_T(t_i)$, $1\le i\le m-1$,
are vertices of different subtrees rooted at a child of $v$ and not containing $\alpha_T(x_0)$.
Observe that the variable $x_0$ does not appear in $t_1$, \ldots, $t_{m-1}$ by (\refclaim{stat1}).

Let $X_0$ be the subset of $X$ consisting of the terms mapped by $\alpha_T$ 
to a vertex of \rt{the unique subtree of $T$ that is rooted at a child of $v$ 
and includes $\alpha_T(x_0)$.}
Note that all terms in $X_0$ contain $x_0$ \rt{by (\refclaim{stat2}) and
the choice of $v$}.
Let $T_0$ be the template obtained from $T$ by taking the minimal rooted subtree
containing $\alpha_T(X_0)$ and the root of the tree containing $\alpha_T(x_0)$, and
restricting the function $\alpha_T$ to the terms containing $x_0$.
Further, let $X'_0$ be the subset of $X$ consisting
the terms $t$ such that $\alpha_T(t)$ lies on the path \rt{of $T$ from a} 
root \rt{to} $v$.
Observe that the construction of $\varphi'$ implies that

{
\newclaim{stat3}{%
if a term from $X_0$ appears in a clause of $\varphi'$, then
every term that appears in that clause belongs to $X_0\cup X'_0$.
}

}

Let $\varphi''\rt{(x_1,x_2,\ldots,x_n)}$ be the formula obtained from 
$\varphi'$ by removing clauses containing at least one term from $X_0$ and
replacing each term $t\in X'_0$ containing $x_0$ with $p^k(t_v)$,
where $k$ is the integer such that $\alpha_T(t)=q^k(\alpha_T(t_v))$.
\rt{It follows from (\refclaim{stat1}) that the variable $x_0$ does not
appear in $\varphi''$.}
Let $T'$ be the template obtained from $T$ by taking the minimal \rt{subforest} containing
all the terms without $x_0$ and restricting the function $\alpha_T$ to such terms.
The formula $\overline{\varphi}_T$ will then be the conjunction of the following formulas:
\begin{itemize}
\item[(a)] the formula $\varphi''(x_1,\ldots,x_n)$,
\item[(b)] the formula $\xi_{T'}$ from Lemma~\ref{lm-embed} \rt{applied to} 
the template $T'$, \rt{the language $\overline L$ and the set of $\overline L$-terms $X\setminus X_0$}, and
\item[(c)] the formulas
      $$\neg\left(\bigwedge_{i\in Y} U_0(p^{k_i-1}(t_i))\right)\lor U_{|Y|+1}(p^k(t_v))$$
      for all subsets $Y$ of the set $\{1,\ldots,m-1\}$
      where $k$ is the integer such that $q^k(\alpha_T(t_v))=v$ and $k_i$, $i=1,\ldots,m-1$,
      are the integers such that $q^{k_i}(\alpha_T(t_i))=v$ (we note here that a conjunction
      over an empty set is true by convention).
\end{itemize}
\rt{This completes the definition of $\overline{\varphi}_T$.}
\bigskip

\rt{%
Let $F$ be a rooted forest of depth at most $d$, and let $S$ be an 
$L$-structure guarded by the closure of $F$.
We need to define an $\overline L$-structure $\overline S$
such that $\overline S$ is an expansion of $S$ and $p^{\overline S}$ is
the $F$-parent function.
To do so we need to define the interpretations 
$U_0^{\overline S}, U_1^{\overline S}, \ldots, U_K^{\overline S}$.
}

We define \rt{the} unary relation $U^{\rt{\overline S}}_0(w)$ to be the set of elements $w$ of $F$ at depth $d_v+1$ such that
the subtree of $w$ in $F$ contains an element $v_0$ at depth $d_{x_0}$ (in $F$) with the following properties:
\begin{itemize}
\item there is a $(v_0)$-admissible embedding of the template $T_0$ in $F$ for $S$, and
\item all clauses appearing in the conjunction $\varphi'$ with at least one term from $X_0$
      are true with $x_0=v_0$ and the terms $t\in X'_0$, say $\alpha_T(t)=q^k(\alpha_T(x_0))$,
      replaced with $(p^{\overline{S}})^k(v_0)$.
\end{itemize}
The relation $U^{\rt{\overline S}}_0(w)$ can be computed as follows: for every element $v_0\in V(S)$ at depth $d_{x_0}$ of $F$,
evaluate all terms in $X_0$ \rt{by substituting $v_0$ for $x_0$} and
testing whether the tree $T_0$ and the rooted subtree of $F$
containing the values of the terms are isomorphic as rooted trees (this can be done in time
linear in the size of $T_0$ which is constant). If they are isomorphic,
evaluate the clauses in the conjunction $\varphi'$ with at least one term from $X_0$
with the terms in $X'_0$ replaced with $(p^{\overline{S}})^k(v_0)$.
If all of them are true, add the ancestor $w$ of $v_0$ at depth $d_v+1$ in $F$ to $U_0$ (note that
$w$ and $v_0$ coincide if their depths are the same).
\rt{This produces a valid result by (\refclaim{stat3}).}
Since the time spent by checking every vertex $v_0$ at depth $d_{x_0}$ of $F$ is constant,
the time needed to compute $U^{\rt{\overline S}}_0$ is linear.

\rt{For $i=1,2,\ldots,K$ we} define \rt{the} unary relation 
$U^{\rt{\overline S}}_i(w)$ to be the set of elements $w$ of $F$
at depth $d_v$ such that $U^{\rt{\overline S}}_0(w')$ is true
for at least $i$ children $w'$ of $w$. Clearly, the relations
$U^{\rt{\overline S}}_i(w)$, $1\le i\le \rt{K}$, can be computed in linear time when the relation $U^{\rt{\overline S}}_0$ has been determined.

We now verify that the formula $\overline{\varphi}_T$ has the desired properties.
\rt{Let $v_1,\allowbreak v_2,\ldots,v_n\in V(S)$.}
Suppose \rt{first} that $\overline{S}\models\overline{\varphi}_T(v_1,\ldots,v_n)$.
Thus $\overline S$ satisfies the formulas listed in (a)--(c) above.
\rt{Since $\overline S\models\xi_{T'}(v_1,v_2,\dots,v_n)$}
there exists a $(v_1,\ldots,v_n)$-admissible embedding \rt{$\nu'$} of 
$T'$ in $F$ \rt{for $\overline S$}.
\rt{%
Let $Y$ be the set of all integers $i\in\{1,2,\ldots,m-1\}$ such that
$U_0^{\overline S}(\nu'(\alpha_{T'}(p^{k_i-1}(t_i) )))$ holds, where
$k_i$ is as in (c) above.
Since $\overline S$ satisfies the formula in (c) corresponding to the
set $Y$ we deduce that the vertex $\nu'(v)$
}
has a son $w$ such that $U^{\rt{\overline S}}_0(w)$ is true and
the subtree of $F$ rooted in $w$ does not contain the value of any term 
in $X\setminus X_0$.
In particular, the subtree rooted in $w$ contains a vertex $v_0$ such that
\rt{$\nu'$}
can be extended
to a $(v_0,\ldots,v_n)$-admissible embedding of $T$ in $F$ for $S$ and
all clauses in the conjunction $\varphi'$ containing a term from $X_0$ are satisfied with $x_0=v_0$.
The clauses of $\varphi'$ not containing a term from $X_0$ appear in $\varphi''$ and
they are satisfied \rt{by $S$} since $\overline{S}\models\varphi''(v_1,\ldots,v_n)$.
\rt{Thus $S\models\varphi'(v_0,v_1,\ldots,v_n)$,  and hence
$S\models\varphi(v_0,v_1,\ldots,v_n)$ by (\refclaim{stat4}).}

On the other hand, \rt{assume that there exists $v_0\in V(S)$ such that}
$S\models\varphi(v_0,\ldots,v_n)$ and there exists
a $(v_0,\ldots,v_n)$-admissible embedding \rt{$\nu$} of $T$ 
\rt{into $F$ for $S$}.
\rt{From (\refclaim{stat4})}
it follows that $S\models\varphi'(v_0,\ldots,v_n)$,
\rt{%
and hence $\overline{S}\models\varphi''(v_1,\ldots,v_n)$.
The restriction of $\nu$ to $T'$ shows that 
$\overline{S}\models\xi_{T'}(v_1,\ldots,v_n)$.}
Let $w$ be the son of $\rt{\nu(v)}$ whose subtree contains $v_0$. 
It follows that $U_{\rt{0}}^{\rt{\overline S}}(w)$.
\rt{The existence of $w$ shows that the formulas listed in item (c) are
satisfied by $\overline S$. Thus $\overline S$ satisfies all formulas in 
(a)--(c), and hence}
it follows that $\overline{S}\models\overline{\varphi}_T(v_1,\ldots,v_n)$.
This completes the proof under the assumption that 
\zz{(\refclaim{stat0}) holds.}

The complementary case \zz{when (\refclaim{stat0}) does not hold}
is handled similarly.
In this case, the predicate $U^{\rt{\overline S}}_0$ is defined for the roots of the trees of $F$, and
the \zz{nullary} predicates $U^{\rt{\overline S}}_1$, \ldots, $U^{\rt{\overline S}}_{K}$ 
are such that such that $U^{\rt{\overline S}}_i$ is true 
if $U^{\rt{\overline S}}_0(r)$ is satisfied for at least $i$
roots \rt{$r$} of the trees in $F$.
\end{proof}

We now prove \rtt{a} lemma that forms the core of our first algorithm.

\begin{lemma}[Quantifier elimination lemma]
\label{lm-elim}
Let $d\ge 0$ be an integer, $L$ a language and $\varphi(x_1,\ldots,x_n)$ 
a simple $L$-formula of
the form $\exists x_0 \; \varphi'(x_0,\ldots,x_n)$ 
such that $\varphi'(x_0,\ldots,x_n)$ is a quantifier-free $L$-formula
with free variables $x_0,\ldots,x_n$.
There exist a language $\overline{L}$  and a quantifier-free
(not necessarily simple) 
$\overline{L}$-formula $\overline{\varphi}$ such that
the following holds:
\begin{itemize}
\item $\overline{L}$ is obtained from $L$ by adding a function symbol $p$
      and finitely many relation symbols of arity  \rt{one}, and
\item for every rooted forest $F$ of depth at most $d$ and every $L$-structure $S$
      guarded by the closure of $F$,
      there exists an $\overline{L}$-structure $\overline{S}$ 
      such that $\overline{S}$ is an expansion of $S$ and
      for every $v_1,\ldots,v_n\in V(S)$,
      \begin{center}
      $S\models\varphi(v_1,\ldots,v_n)$ if and only if $\overline{S}\models\overline{\varphi}(v_1,\ldots,v_n)$
      \end{center}
      where $p^{\overline{S}}$ is the $F$-parent function and
      the \rt{interpretations in $\overline S$ of} the new relation symbols
      can be computed (by listing the singletons they contain) in linear time given $F$ and $S$.
\end{itemize}
\end{lemma}

\begin{proof}
Let $d$, $L$ and $\varphi'$ be fixed.
We assume without loss of generality that
the formula $\varphi'$ is in the disjunctive normal form and
all the variables $x_0,\ldots,x_n$ appear in $\varphi'$.
Let $F$ be a rooted forest of depth at most $d$, and let $S$ be
an $L$-structure.

The proof proceeds by induction on the length of $\varphi'$. If $\varphi'$ is
a disjunction of two or more conjunctions, i.e., $\varphi'=\varphi_1\lor\varphi_2$,
we apply induction to the formulas $\exists x_0\varphi_1$ and $\exists x_0\varphi_2$.
We obtain languages $L_1$ and $L_2$, and for $i=1,2$ 
an $L_i$-formula $\overline{\varphi}_i$ and 
an $L_i$-structure $\overline{S}_i$.
We assume that the new unary relation symbols of $L_1$ and $L_2$ are distinct and
set $\overline{L}^r=L^r_1\cup L^r_2$, $\overline{L}^f=L^f_1=L^f_2=L^f\cup\{p\}$ and
$\overline{\varphi}=\overline{\varphi}_1\lor\overline{\varphi}_2$.
We define the $\overline{L}$-structure $\overline{S}$ by 
$V(\overline{S})=V(S)$ and by taking the interpretations of symbols
from $\overline{S}_1$ and $\overline{S}_2$.

Thus in the remainder of the proof we may assume that 
$\varphi'$ is a conjunction.
Let $v_1,v_2,\ldots,v_n\in V(S)$.
By Lemma~\ref{lm-extemp} we have
$S\models\varphi(v_1,\ldots,v_n)$ if and only if
there exist $v_0\in V(S)$ and a $\varphi'$-template $T$ of depth at 
most $d$ such that $S\models\varphi'(v_0,\ldots,v_n)$ and there
exists an embedding
of $T$ into $F$ that is $(v_0,\ldots,v_n)$-admissible for $S$.
By Proposition~\ref{prop-fin-template} the number of $\varphi'$-templates
of depth at most $d$ is bounded by a function of $\varphi$ and $d$.
By Lemma~\ref{lm-phi-T},
for every $\varphi'$-template $T$ of depth at most $d$,
there exist a language $L_T$, a quantifier-free ${L_T}$-formula 
${\varphi}_T$ and an $L_T$-structure $S_T$ that is an expansion of $S$
such that for every $v_1,\ldots,v_n\in V(S)$,
$S_T\models{\varphi}_T(v_1,\ldots,v_n)$ if and only if
there exists $v_0$ such that
there is a $(v_0,v_1,\ldots,v_n)$-admissible embedding of $T$ in $F$ for $S$ and
$S\models\varphi'(v_0,v_1,\ldots,v_n)$.
We may assume that for distinct $\varphi'$-templates $T$ and $T'$,
if a function or a relation symbol belongs both $L_T$ and $L_{T'}$,
then it belongs to $L$.
Let $\overline{L}$ be the language consisting of all function and relation
symbols of all $L_T$,
let the formula $\overline{\varphi}$ be obtained as
the disjunction of the $\overline{L}$-formulas $\overline{\varphi}_T$,
where the disjunction runs over all choices of $\varphi'$-templates $T$,
and let the $\overline{L}$-structure $\overline{S}$ be obtained 
by taking the union of the interpretations of all $S_T$.
Then $\overline{L}$, $\overline{\varphi}$ and $\overline{S}$
are as desired.
\end{proof}

In order to apply Lemma~\ref{lm-elim}, the \rt{given} formula needs to be simple
but the lemma produces a formula that need not be simple.
The following lemma copes with this issue.

\begin{lemma}
\label{lm-simple}
\rt{Let $\GG$ be a class of graphs of bounded expansion,}
$L$  a language and $\varphi(x_1,\ldots,x_n)$ an $L$-formula with $q$ quantifiers.
There exist \rt{a class $\GG'$ of graphs of bounded expansion,}
a language $L'$ that extends $L$,
and a simple $L'$-formula $\varphi'(x_1,\ldots,x_n)$ with $q$ quantifiers 
with the following properties.
For every $L$-structure $A$ guarded by a graph $G\rt{\in\GG}$,
there exists an $L'$-structure $A'$ guarded
by a graph \rt{$G'\in\GG'$}
such that \rt{$V(G)=V(G')$, $A'$ is an expansion of $A$ and}
$A\models\varphi(v_1,\ldots,v_n)$ if and only if $A'\models\varphi'(v_1,\ldots,v_n)$
for any $v_1,\ldots,v_n\in V(A)=V(A')$.
Moreover, an $L'$-structure $A'$ and graph $G'$
satisfying the above specifications
can be computed in time \rt{$O(|V(G)|)$}.
\end{lemma}

\begin{proof}
We may assume that $\varphi$ is not simple, for otherwise there is
nothing to prove.
Let $f$ and $g$ be function symbols of $L$ such that the $L$-term $g(f(t))$
appears in $\varphi$ for some $L$-term $t$.
Let \rt{$\GG_1$ be the class of all first augmentations of members of $\GG$;
then $\GG_1$ has bounded expansion by Theorem~\ref{thm:augmexp}.}
Let $L_1$ be the extension of $L$ obtained by adding a new function
symbol $h$, and for an $L$-structure $A$ we define an $L_1$-structure
$A_1$ as the expansion of $A$, where the interpretation of $h$ is
defined by $h^{A_1}(v)=g^A(f^A(v))$ for all $v\in V(A)$.
Let $\varphi_1$ be obtained from $\varphi$ by replacing \rt{all} appearances
of $g(f(t))$ by $h(t)$.
Then clearly $A\models\varphi(v_1,\ldots,v_n)$ if and only if 
$A_1\models\varphi'(v_1,\ldots,v_n)$
for all $v_1,\ldots,v_n\in V(A)=V(A_1)$.
Let $D'$ be an orientation of $G$ of maximum in-degree $2\nabla_0(G)$,
and let $D$ be obtained from $D'$ by adding all directed edges
with head $v$ and tail $f^{\rt{A}}(v)$ and all directed edges with head $v$
and tail $g^{\rt{A}}(v)$. Since the orientation $D'$ can be obtained in a greedy way,
this step can be performed in time $O(|V(G)|+|E(G)|)\rt{=O(|V(G)|)}$.
Let $G_1$ be the augmentation of $D$.
Then $G_1$ is a first augmentation of $G$ (here we \rt{make} use of the term ``+2''
in the definition of an augmentation) and $A_1$ is guarded by $G_1$.
By repeating this construction at most $k$ times, where $k$ is the maximum
number of function compositions appearing in $\varphi$,
we arrive at a desired formula $\varphi'$. Since each step requires
linear time,
the total running time is linear,
as desired.
\end{proof}

We are now ready to prove Theorem~\ref{thm-bounded2};
we prove it in a stronger form needed in Section~\ref{sec-locally}.

\begin{theorem}
\label{thm-bounded+}
Let $\GG$ be a class of graphs with bounded expansion,
$L$ a language and $\varphi(x_1,\ldots,x_n)$ an $L$-formula.
There exist a language $\overline L$, 
\rt{class $\overline\GG$ of graphs with bounded expansion,
a quantifier-free $\overline L$-formula $\overline\varphi(x_1,\ldots,x_n)$}
 and an algorithm \rt{$\cal A$ such that the following holds.}
Given an $L$-structure $A$ guarded by \rt{a graph}
$G\in\GG$ \rt{the algorithm $\cal A$ finds a graph 
$\overline G\in\overline\GG$ with $V(G)=V(\overline G)$}
and an $\overline L$-structure $\overline A$ guarded by $\overline G$ such that
\rt{$V(\overline A)=V(A)$ and}
for all $v_1,\ldots,v_n\in V(A)=V(\overline A)$
\begin{center}
$A\models\varphi(v_1,\ldots,v_n)$ if and only if 
$\overline A\models\overline\varphi(v_1,\ldots,v_n)$.
\end{center}
The running time of the algorithm \rtt{$\cal A$} is $O(|V(G)|)$.
In particular, if $n=0$,
the algorithm decides whether $A\models\varphi$.
\end{theorem}

\begin{proof}
It suffices to show \rt{the existence of 
$\overline L$, $\overline \GG$, $\overline\varphi$ and $\cal A$} 
satisfying the
specifications of the theorem, except that 
\rt{rather than being quantifier-free,} $\overline\varphi$ has one
fewer quantifier than $\varphi$.
A proof of the theorem is then obtained by iterating this \rt{argument}.

If $\varphi$ is quantifier-free, then there is nothing to prove. Hence, 
\rt{we may and will} assume that
$\varphi$ contains at least one quantifier.
By 
Lemma~\ref{lm-simple} we may  assume that $\varphi$ is simple.

Since $\forall x\;\psi$ is equivalent to $\neg\exists x\;\neg\psi$, we can assume that $\varphi$ contains
a subformula \rt{$\xi(x_1,x_2,\ldots,x_N)$} of the form 
$\exists x_0 \psi\rt{(x_0,x_1,\ldots,x_N)}$,
where $\psi$ is a formula with variables $x_0,x_1,\ldots,x_N$ 
and with no quantifiers.
\rt{We will define a desired formula by replacing the subformula $\xi$ of 
$\varphi$ by a different formula.}


Let $X$ be the set of all $L$-terms that apear in $\xi$,  let
$k:=3(|X|^2+1)^2$, and let $\overline\GG$ be the class of all $k$-th augmentations
of members of $\GG$. By Theorem~\ref{thm:augmexp} the class $\overline\GG$
has bounded expansion. By Lemma~\ref{lem:bound} there exists an
integer $K$ such that every member of $\overline\GG$ is $K$-colorable.
Let $L'$ be the langauge obtained from $L$ by adding $K$ unary relation
symbols $C_1,C_2,\ldots,C_K$, 
\rtt{and for each function symbol $f$ of $L$ another $K$ unary relation
symbols $C_{f,1},C_{f,2},\ldots,C_{f,K}$.}
(Their interpretations in a structure $A$ will be used to encode a given 
$K$-coloring of the graph guarding $A$.)
Let $\Lambda$ be the set of all mappings $X\to\{1,2,\ldots,K\}$,
\rtt{and let} $\alpha\in\Lambda$.
\rtt{For a term $t\in X$ of the form $x_i$ let $E_t:=C_{\alpha(t)}(t)$,
and for $t\in X$ of the form $f(x_i)$ let $E_t:=C_{f,\alpha(t)}(x_i)$.}
Let $\varphi_\alpha$ denote the $L'$-formula
$\bigwedge_{t\in X} E_t$, and let
$\xi_\alpha(x_1,x_2,\ldots,x_N)$ denote the formula
$\exists x_0 (\psi(x_0,x_1,\ldots,x_N)\wedge\varphi_\alpha(x_0,x_1,\ldots,x_N))$.
Let $\overline L_\alpha$ and $\overline \xi_\alpha$ be a language and
a formula obtained by applying Lemma~\ref{lm-elim} to the language $L'$ and the formula
$\xi_\alpha$.
Finally, let $\overline L$ be the language obtained by taking the
union of all function and relation symbols of $L$ and all the languages
$\overline L_\alpha$, and let $\overline\varphi$ be the $\overline L$-formula
obtained from $\varphi$ by replacing the subformula $\xi$ of $\varphi$
by the disjunction of $\overline\xi_\alpha$ over all $\alpha\in\Lambda$.
We will show that $\overline\varphi$ is as desired.

To prove this
let $G\in\GG$, let $A$ be an $L$-structure guarded by $G$,
let $\overline G\in\overline\GG$ be a $k$-th augmentation of $G$,
let $D_1,D_2,\ldots,D_k$ be as in the defintion of $k$-th augmentation,
and let $c$ be a $K$-coloring of $\overline G$.
The coloring $c$ exists by our choice of $K$,
and $\overline G, D_1,D_2,\ldots,D_k$ and $c$ can be computed in linear
time by Lemmas~\ref{lem:bound} and~\ref{alg:augment}.
Let $A'$ be the $L'$-structure defined by saying that it is an expansion
of $A$,  that $C^{A'}_i$ consists of all $v\in V(A')$ such that
$c(v)=i$,
\rtt{and that $C^{A'}_{f,i}$ consists of all $v\in V(A')$ such that
$c(f^{A'}(v))=i$}. 

\claimno=0

{
\newclaim{cl1}{ 
For all $v_1,v_2,\ldots v_N\in V(A)$ we have
$A\models\xi(v_1,v_2,\ldots v_N)$
if  and only if there exists $\alpha\in\Lambda$ such that
$A'\models \xi_\alpha(v_1,v_2,\ldots v_N)$.}

}

\noindent
To prove (\refclaim{cl1}) we note that the ``if" part is clear.
To prove the ``only if" part let $A\models\xi(v_1,v_2,\ldots v_N)$.
Thus there exists $v_0\in V(A)$ such that $A\models\psi(v_0,v_1,\ldots v_N)$.
Let $t\in X$. If $t$ is a variable $x_i$, then let $\alpha(t):=c(v_i)$,
and if $t$ is of the form $f(x_i)$ for a function symbol $f$, then
let $\alpha(t):=c(f^A(v_i))$.
Then $A'\models\xi_\alpha(v_1,v_2,\ldots v_N)$, as desired.
This proves (\refclaim{cl1}).
\medskip

For $\alpha\in\Lambda$ let $A_\alpha$ be an $L'$-structure defined
as follows. 
We let $V(A_\alpha)$ be the set of all $v\in V(A)$ such that 
$c(v)\in\alpha(X)$.
For a function symbol $f$ in the language $L'$ let
$f^{A_\alpha}(v):=f^A(v)$ if both $c(v)$ and $c(f^A(v))$ belong to
$\alpha(X)$, and let $f^{A_\alpha}(v):=v$ otherwise.
For a relation symbol $R$ in $L'$ of arity $l$ let $R^{A_\alpha}$
be the subset of $R^A$ consisting of all $l$-tuples whose every element
belongs to $\alpha(X)$.

{
\newclaim{cl2}{
For all $\alpha\in\Lambda$ and all $v_1,v_2,\ldots v_N\in V(A)$ we have
$A'\models\xi_\alpha(v_1,v_2,\ldots v_N)$
if  and only if 
$A_\alpha\models \xi_\alpha(v_1,v_2,\ldots v_N)$.}

}

\noindent
\rtt{%
To prove (\refclaim{cl2}) we first notice that 
$A'\models\varphi_\alpha(v_0,v_1,\ldots v_N)$ 
and 
$A_\alpha\models\varphi_\alpha(v_0,v_1,\ldots v_N)$ 
are both equivalent to
$c(v_i)=\alpha(x_i)$ for every $L$-term in $X$ of the form $x_i$
and $c(f^{A'}(v_i))=\alpha(f(x_i))$ for every $L$-term in $X$ of the form $f(x_i)$,
in which case $c(v_i)\in\alpha(X)$ for every $L$-term in $X$ of the form $x_i$
and $f^{A'}(v_i)=f^{A_{\alpha}}(v_i)$ and $c(f^{A'}(v_i))\in\alpha(X)$
for every $L$-term in $X$ of the form $f(x_i)$.
So, $A'\models\varphi_\alpha(v_0,v_1,\ldots v_N)$ if and only if
$A_{\alpha}\models\varphi_\alpha(v_0,v_1,\ldots v_N)$ (note that
if $f^{A'}(v_i)\not=f^{A_{\alpha}}(v_i)$ for some $i$,
then $E_{f(x_i)}$ fails for both $A'$ and $A_\alpha$).
We deduce that (\refclaim{cl2}) holds.
\medskip
}

For $\alpha\in\Lambda$ let $H_\alpha$ be the subgraph of $G$ induced
by vertices $v$ such that $c(v)\in\alpha(X)$.
By Theorem~\ref{lem:lowtdepth} there exists a rooted forest $F_\alpha$
of depth at most $|X|$ such that $H_\alpha$ is a subgraph of the closure
of $F_\alpha$ and the \rtt{closure} of $F_\alpha$ is
a subgraph of $\overline G$.
\rtt{Thus $A_{\alpha}$ is guarded by the closure of $F_\alpha$.}
By Theorem~\ref{alg:lowtdepth} the rooted forest $F_\alpha$ can be
found in linear time, because $|E(G')|=O(|V(G)|)$ by Theorem~\ref{thm:augment}.
Let $\overline A_\alpha$ be an $\overline L_\alpha$-structure
as in Lemma~\ref{lm-elim} applied to the $L'$-structure $A_\alpha$
and rooted forest $F_\alpha$. 
\rtt{Then $\overline A_\alpha$ is guarded by the closure of $F_\alpha$}
and

{
\newclaim{cl3}{ 
for all $\alpha\in\Lambda$ and all $v_1,v_2,\ldots v_n\in V(A)$ we have
$A_\alpha\models\xi_\alpha(v_1,v_2,\ldots v_n)$
if  and only if 
$\overline A_\alpha\models \overline\xi_\alpha(v_1,v_2,\ldots v_n)$.}

}

Let $\overline A$ be an $\overline L$-structure defined as follows.
Let $f$ be a function symbol from $\overline L$. If $f$ belongs to $L$,
then $f^{\overline A}(v):=f^A(v)$, and if $f$ belongs to $\overline L_\alpha$,
then $f^{\overline A}(v):=f^{\overline A_\alpha}(v)$.
We define the interpretations of relation symbols analogously.
\rtt{Since $\overline A_\alpha$ is guarded by the closure of $F_\alpha$
and the closure of $F_\alpha$ is a subgraph of
$\overline G$, we deduce that}
$\overline A$ is guarded by $\overline G$.

{
\newclaim{cl4}{
For all $v_1,v_2,\ldots v_n\in V(A)$ we have
$\overline A\models \overline\xi(v_1,v_2,\ldots v_n)$
if  and only if there exists $\alpha\in\Lambda$ such that
$\overline A_\alpha\models \overline\xi_\alpha(v_1,v_2,\ldots v_n)$.}

}

\noindent
The proof of (\refclaim{cl4}) is clear.
\medskip

It follows from claims (\refclaim{cl1})--(\refclaim{cl4}) that $\overline A$
is as desired, and the construction shows that it can be computed
from $G$ and $A$ in time $O(|V(G)|)$.
\end{proof}

\section{Deciding FO properties in graphs with locally bounded expansion}
\label{sec-locally}

The following theorem uses a result of Gaifman~\cite{gaif} that
FO properties are local in a certain sense. 
The theorem is implicit in~\cite{fg} (see also~\cite{grokremeth}).

\begin{theorem}
\label{thm-local-general}
Let $\GG$ be a class of graphs and for an integer $d\ge0$ let
$\GG_d$ be the class of graphs
consisting of all induced subgraphs of $d$-neighborhoods of graphs in $\GG$.
Let $\GG_d$ \rtt{have bounded expansion} for all integers $d\ge0$.
Furthermore, let $L$ be a language and $L'$ the language obtained from $L$
by adding a new binary relation symbol.
Suppose that
for every $d$ and every $L'$-formula $\varphi'(x)$,
there exists a linear-time algorithm that
lists all elements $v$ of an input $L'$-structure guarded by a graph from $\GG_d$
that satisfy $\varphi'(v)$.
Then, for every $L$-sentence $\varphi$
there exists an almost linear-time algorithm that
decides whether an input $L$-structure guarded by a graph from $\GG$
satisfies $\varphi$.
\end{theorem}

\begin{proof}
\rt{We show how to modify the proof of~\cite[Theorem~1.2]{fg} 
to yield a proof of this theorem.}
%
The proof of~\cite[Theorem~1.2]{fg} 
relies on Lemma~4.4, Corollary~6.3, Corollary~8.2 and Lemma~8.3 
from the same paper, and those assume that $\GG$ has ``bounded local
tree-width". \rt{In our context} Lemma 4.4 would be needed to 
justify that \rt{for every $L$-formula $\psi(x)$ and every integer $d\ge0$
there exists a linear-time algorithm that given an $L$-structure $A$
guarded by a member of $\GG_d$ computes the set of all $v\in V(A)$
such that $A\models\psi(v)$.
This follows from the hypothesis of the theorem instead.
(Here we do not need the extension $L'$.)
}


Corollary~6.3 \rtt{and Corollary~8.2}
apply in our setting without any alterations \rtt{with the same proofs,
using the fact that for every fixed integer $d$ every graph $G\in\GG_d$ 
has at most $O(|V(G)|)$ edges by Lemma~\ref{lem:bound}.}

Finally, in Lemma~8.3, 
\rt{%
we need to be able to compute, in linear time for every fixed $r$, given
an $L$-structure $A$ and $v\in V(A)$, the set of elements of $V(A)$
at distance at most $r$ in the Gaifman graph of $A$.
This can be derived by applying the hypothesis of the theorem to
the $L'$-structure $A'$, where $A'$ is the expansion of $A$ defined
by saying that the interpretation of the new binary relation is
adjacency in the Gaifman graph of $A$.
This relation can be computed in linear time.
To carry out the last step of the algorithm of Lemma~8.3 we apply
Theorem~\ref{thm-bounded2}.
}
\end{proof}

\begin{proof}[Proof of Corollary~\ref{thm-nowhere2}]
Let $\GG$, $L$ and $\varphi$ be as in Corollary~\ref{thm-nowhere2},
\rt{and let $\GG_d$ be as in Theorem~\ref{thm-local-general}}.
In particular, the class $\GG_d$ has bounded expansion for every $d$. By Theorem~\ref{thm-bounded+},
for every integer $d$ and every $L'$-formula $\varphi'(x)$, there exist a language $L''$ and
a quantifier-free $L''$-formula $\varphi''(x)$ such that every $L'$-structure $A$
guarded by a graph from $\GG_d$ can be transformed in linear time to an $L''$-structure $A'$
with $V(A)=V(A')$ such that $A\models\varphi'(v)$ if and only if $A'\models\varphi''(v)$ for every $v\in V(A)$.
In particular, it is possible to list in linear time all $v\in V(A)$ such that $A'\models\varphi''(v)$
since evaluating the latter formula requires constant time.
So, the assumptions of Theorem~\ref{thm-local-general} are satisfied.
\end{proof}

\section{Dynamic data structures for $\Sigma_1$-queries}
\label{sec-subgraph}

In this section, we provide two data structures for answering $\Sigma_1$-queries.
The update time is constant but the price we have to pay is that the graph that guards
the relational structure must be fixed before the computation starts.
Before we start our exposition, we need to introduce more definitions.

Let $L$ be a language with no function symbols.
For an integer $k\ge1$, a {\em $k$-labelled} $L$-structure
is a pair $(S,\sigma)$, where $S$ is an $L$-structure and $\sigma$
is an injective mapping dom$(\sigma)\to V(S)$, where dom$(\sigma)\subseteq\{1,2,\ldots,k-1\}$.

The {\em trunk} of a $k$-labelled $L$-structure $(S,\sigma)$
is the $k$-labelled $L$-structure $(S',\sigma)$, where $S'$ is obtained 
from $S$ by removing
all tuples $(v_1,\ldots,v_t)$ with $v_1,\ldots,v_t\in\hbox{\rm dom}(\sigma)$
from each relation of $S$.
A $k$-labelled $L$-structure $(S,\sigma)$ is {\em hollow}
if it is equal to its trunk.
Two $k$-labelled $L$-structures $(S_1,\sigma_1)$ and $(S_2,\sigma_2)$ 
are {\em $k$-isomorphic} if $\hbox{\rm dom}(\sigma_1)=\hbox{\rm dom}(\sigma_2)$ 
and their trunks are isomorphic by way of an isomorphism $f:V(S_1)\to V(S_2)$
such that $\sigma_2(i)=f(\sigma_1(i))$ for every $i\in\hbox{\rm dom}(\sigma_1)$.
In particular, every $k$-labelled $L$-structure is $k$-isomorphic to its trunk.

Suppose now that an $L$-structure $S$ is guarded by the closure of a 
rooted tree $T$.
For a vertex $v$ of $T$ at depth $d$, let $P_T(v)$ denote the vertex-set
of the path from the root of $T$ to $v$ and
$T\langle v\rangle$ the vertex-set of the subtree of $v$ (including $v$ itself).
Then, $S\langle v\rangle$ denotes the set of all $d$-labelled $L$-structures 
$(S',\sigma)$ such that
$S'$ is an induced substructure of $S$ with elements only in 
$P_T(v)\cup T\langle v\rangle$ and
dom$(\sigma)$ consists of all integers $i\in\{1,2,\ldots,d-1\}$ such that
$V(S')$ includes an element at depth $i$, in which case $\sigma(i)$ is
equal to that element.

We are now ready to prove a lemma that contains the core of our data structure.

\begin{lemma}
\label{lm-databe}
Let $L$ be a language with no function symbols, $d_0$ a fixed integer and $F$ a rooted forest of depth at most $d_0$.
There exists a data structure representing an $L$-structure $S$ guarded by the closure of $F$ such that
\begin{itemize}
\item the data structure is initialized in linear time,
\item the data structure representing an $L$-structure $S$ can be changed
      to the one representing an $L$-structure $S'$ by adding or removing a tuple
      from one of the relations in constant time provided that $S'$ is guarded by the closure of $F$, and
\item the data structure decides in time bounded by $O(|\varphi|)$
      whether a given $\Sigma_1$-$L$-sentence $\varphi$ with at most $d_0$ variables is satisfied by $S$,
      and if so, it outputs one of the satisfying assignments.
\end{itemize}
\end{lemma}

\begin{proof}
For every vertex $v$ of $F$ at depth $d$,
we will store the following two lists:
\begin{itemize}
\item for every relation symbol $R$ of $L$
      the list of all tuples $\tau\in R^S$ such that $\tau$ includes $v$ and
      all elements of $\tau$ belong to $P_T(v)$, where $T$ is the tree of $F$ containing $v$, and
\item the list of all (non-$d$-isomorphic) $d$-labelled hollow $L$-structures with at most $d_0$ elements
      that are $d$-isomorphic to a $d$-labelled $L$-structure contained in $S\langle v\rangle$.
\end{itemize}
Since there are only finitely many non-$d$-isomorphic $d$-labelled
$L$-structures with at most $d_0$ elements for every $d\le d_0$,
the length of each list of the second type is bounded by a constant depending
only on $d_0$ and $L$.
If $v$ is a non-leaf vertex of $F$, there will be a third list associated with $v$:
\begin{itemize}
\item the list of all (non-isomorphic) $(d+1)$-labelled hollow $L$-structures 
      $(S',\sigma)$
      with at most $d_0$ elements that are isomorphic to a member of the 
      second list of at least one child of $v$;
      for each such $(S',\sigma)$, there will be stored the list of all children of $v$
      whose second list contains a member isomorphic to $(S',\sigma)$.
\end{itemize}
In addition, there will be a global list of all (non-isomorphic) $L$-structures
with at most $d_0$ elements that appear as induced $L$-substructures in $S$.

Let us describe how all these lists are initialized.
The initialization of the first type of list is trivial:
just put each tuple contained in one of the  relations to the list of its element that is farthest from the root.
This can clearly be done in constant time per tuple.

Initialization of other types of lists is more difficult. Fix a tree $T$ of $F$.
We proceed from the leaves towards the root of $T$. Let $v$ be a vertex of $T$ at depth $d$.
If $v$ is a leaf of $T$ at depth $d$, then the second list of $v$ contains only those hollow
$d$-labelled $L$-structures $(S',\sigma)$ with $V(S')\subseteq P_T(v)$ such that
if $v\in V(S')$, then $S'$ contains relations with their tuples from $S$
containing $v$ and elements from $V(S')$, and
if $v\not\in V(S')$, then all relations of $S'$ are empty.
This can be done in linear time as for each tuple in every relation
one determines whether its element of the largest depth is a leaf and,
if so, it includes the tuple to the structures at that leaf.

Suppose now that $v$ is not a leaf of $T$. The third list associated with $v$
can be initialized by merging the second type of lists of children of $v$. (This needs
time linear in the number of the children, but the sum of the numbers of children of all
vertices is linear in $|T|$.
This will require linear time for the whole structure since the number of
non-$d$-isomorphic $d$-labelled hollow $L$-structures with at most $d_0$ elements
is bounded, and thus the size of each list of the second type is bounded.)
We next describe how it can be decided whether a $d$-labelled hollow 
$L$-structure $(S',\sigma)$
should be contained in the list of $v$ of the second type.
Assume that $S\langle v\rangle$ contains a $d$-labelled hollow 
$L$-structure $(S'',\sigma'')$
that is $d$-isomorphic to $(S',\sigma)$.

Then $V(S'')$ can be decomposed into disjoint subsets $V_0,V_1,\ldots,V_m$ such that
$V_0=V(S'')\cap P_T(v)$, each of the sets $V_i$, $i=1,\ldots,m$,
is fully contained in a subtree of a child $v_i$ of $v$, and
different subsets $V_1,\ldots,V_m$ are contained in different subtrees. 
Observe that each tuple of a relation of $S''$ has its elements in $V_0\cup V_i$
for some $i=1,\ldots,m$. Moreover,
the only tuples in such relations with all elements from $V_0$ are those that contain $v$.

Hence, the existence of $(S'',\sigma'')$ can be tested by
considering all partitions of $V(S')$ into disjoint subsets $V_0,V_1,\ldots,V_m$
such that $\sigma(\hbox{\rm dom}(\sigma))\subseteq V_0$, 
$|V_0\setminus\sigma(\hbox{\rm dom}(\sigma))|\le 1$,
every tuple in a relation of $S'$ has its elements in
$V_0\cup V_i$ for some $i=1,\ldots,m$, and two additional conditions
are satisfied.
To state those conditions let $i\in\{1,2,\ldots,m\}$ and let us
define $(S_i,\sigma_i)$ to be the $(d+1)$-labelled hollow $L$-structure
such that $S_i$ is the substructure of $S$ induced by $V_0\cup V_i$,
dom$(\sigma_i)=\hbox{\rm dom}(\sigma)\cup\{d\}$ if 
$V_0\setminus\sigma(\hbox{\rm dom}(\sigma))\ne\emptyset$ and
dom$(\sigma_i)=\hbox{\rm dom}(\sigma)$ otherwise,
$\sigma_i(j)=\sigma(j)$ for every $j\in\hbox{\rm dom}(\sigma)$
and $\sigma_i(d)$ is the unique element of 
$V_0\setminus\sigma(\hbox{\rm dom}(\sigma))$ if the latter set is not empty.
The two remaining conditions are that
there exist distinct children $v_1,\ldots,v_m$ of $v$ such that the second 
list of $v_i$ has a member isomorphic to $(S_i,\sigma_i)$ and
\zz{that} for each relation symbol $R$ of $L$
the tuples in  $R^{S'}$ containing $\sigma(d)$
are precisely the tuples listed in the first list for $R$ and $v$.

We now describe how to test the existence of children $v_1,\ldots,v_m$.
Let $W$ be the set of children of $v$ such that
\zz{for all $i\in\{1,2,\ldots,m\}$}: if $v$ has at most $m$ children
with their second list containing a $(d+1)$-labelled hollow $L$-structure
$(d+1)$-isomorphic to $(S_i,\sigma_i)$,
then $W$ contains all such children of $v$ (here, we use the lists of the third type).
If $v$ has more than $m$ such children,
then $W$ contains arbitrary $m$ of these children. Clearly, $|W|\le m^2\le d_0^2$.
In order to test the existence of such children $v_1,\ldots,v_m$ of $v$,
we form an auxiliary bipartite subgraph $B$: one part of $B$ is formed
by the numbers $1,\ldots,m$ and the other part by children of $v$ contained in $W$.
A child $w\in W$ is joined to a number $i$ if the second list of $w$ contains
a $(d+1)$-labelled hollow $L$-structure $(d+1)$-isomorphic
to $(S_i,\sigma_i)$.

If $B$ has a matching of size $m$,
then this matching determines the choice of children $v_1,\ldots,v_m$.
On the other hand, if such children exist, $B$ contains a matching of size $m$:
indeed, if $v_i\in W$, then $i$ is matched with $v_i$, and if
$v_i\not\in W$, then $v$ has at least $m$ children whose second
list contains a $(d+1)$-labelled hollow $L$-structure $(d+1)$-isomorphic
to $(S_i,\sigma_i)$, in which case $i$ may be matched with one of those
children that is not matched with any $i'<i$.

Since the order of $B$ is at most $m^2+m$ and the number of disjoint non-empty
partitions of $V(S')$ to $V_0,\ldots,V_m$ is bounded, testing the existence of
a $d$-labelled hollow $L$-structure $S''$ can be performed in constant time for $v$.

It remains to construct the global list
containing $L$-structures $S_0$ with at most $d_0$ elements that
are isomorphic to an induced substructure of $S$.
We proceed similarly as when determining the lists of inner elements of the forest $F$.
For every $L$-structure $S'$ with at most $d_0$ elements,
we compute the list of trees of $F$ that contain $S'$,
i.e., $S'$ is contained in the second list of the root of $F$.
Now, $S_0$ is an induced substructure of $S'$ if and only if there exist element-disjoint
$L$-structures $S_1',\ldots,S_m'$ such that $V(S_0)=V(S_1')\cup\cdots\cup V(S_m')$ and 
$S_1',\ldots,S_m'$
appear in $m$ mutually distinct trees of $F$. For each such partition of $S_0$ into $S_1',\ldots,S_m'$,
we can test whether $S_1',\ldots,S_m'$ appear in the list of roots of $m$ distinct trees of $F$
using the auxiliary bipartite graph described earlier. Since all structures involved
contain at most $d_0$ elements, this phase requires time linear in the number of trees of $F$.

We have shown that the data structure can be initialized in linear time. Let us now focus on updating
the structure and answering queries.
Consider a tuple $(v_1,\ldots,v_k)$ that is added to a relation $R^S$ or removed from a relation $R^S$.
Let $r$ be the root of the tree $T$ in $F$ that contains all the elements $v_1,\ldots,v_k$ and
assume that $v_1,\ldots,v_k$ appear in this order on a path from $r$.
By the definition, the only lists affected by the change are those associated with vertices
on the path $P_T(v_k)$. Recomputing each of these lists requires constant time (we proceed
in the same way as in the initialization phase except we do not have to run
 through the children
of the vertices on the path to determine which of them contain particular $k$-labelled hollow
$L$-substructure $S'$ in their lists).
Since the number of vertices on the path $P_T(v_k)$ is at most $d_0$,
updating the data structure requires constant time only.

It remains to describe how queries are answered.
Let $\varphi$ be a $\Sigma_1$-sentence with $d\le d_0$ variables.
We generate all possible $L$-structures $S_0$ with $|V(S_0)|\zz{\le} d$ and
check whether they satisfy the formula $\varphi$.
Let $\SS_0$ be the set of those that satisfy $\varphi$.
The set $\SS_0$ can be generated in time $O(|\varphi|)$ since $L$ and $d_0$ are fixed.

Observe that $S$ satisfies $\varphi$ if and only if
it has an induced substructure isomorphic to a structure in $\SS_0$ (here,
we use that $L$ has no function symbols).
This can be tested in constant time by inspecting the global list.
Providing the satisfying assignment can be done in constant time
if during the computation for each substructure we store a certificate why it was included
in the list (which requires constant time overhead only).
\end{proof}

We are now ready to describe the data structures.
We start with the one for graphs with bounded expansion.

\begin{theorem}
\label{thm-databe}
Let $L$ be a language with no function symbols, $d_0$ a fixed integer and $\GG$ a class of graphs
with bounded expansion. There exists a data structure representing an
$L$-structure $S$ \rt{guarded by a member of $\GG$} such that
\begin{itemize}
\item given a graph $G\in\GG$ \rt{and an $L$-structure $S$ guarded by $G$}, 
      the data structure is initialized in linear time,
\item \rt{if an $L$-structure $S'$ is obtained from $S$}
      by adding or removing a tuple from one of the relations, then 
      the data structure representing $S$ can be changed
      to the one representing $S'$ 
      in constant time provided that both $S$ and $S'$ are guarded by $G$, and
\item the data structure \rt{allows testing} in time bounded by $O(|\varphi|)$
      whether a given $\Sigma_1$-$L$-sentence $\varphi$ with at most 
      $d_0$ variables is satisfied by $S$,
      and if so, \rt{outputting} one of the satisfying assignments.
\end{itemize}
\end{theorem}

\begin{proof}
\rt{%
Let $k:=3(d_0+1)^2$, and let $\GG'$ be the class of all $k$-th augmentations
of members of $\GG$.
Then $\GG'$ has bounded expansion by Theorem~\ref{thm:augmexp}.
Let $K$ be as in Lemma~\ref{lem:bound} applied to $\GG'$.
Thus $K$ depends only on $\GG$ and $d_0$.
Given $G\in\GG$ we compute, in linear time using Lemma~\ref{alg:augment}, 
a $k$-th augmentation $G'$
of $G$ and directed graphs $D_1,D_2,\ldots,D_k$ as in the definition
of $k$-th augmentation.
Then we compute a $K$-coloring $c$ of $G'$ in linear time by Lemma~\ref{lem:bound}.
Let $\cal X$ be the set of all subsets of $\{1,2,\ldots,K\}$ of size $d_0$.
By Theorem~\ref{lem:lowtdepth} $c$ is a $G'$-compliant low tree-depth
coloring of $G$ of order $d_0$,
and by  Theorem~\ref{alg:lowtdepth} we can find in linear time,
for each $X\in\cal X$, a rooted forest $F_X$ such that the subgraph $H_X$
of $G$ induced by vertices $v$ with $c(v)\in X$ is a subgraph of
the closure of $F_X$. 
Now given $X\in\cal X$ and an $L$-structure $S$ guarded by $G$,
let $S_X$ denote the induced substructure  of $S$ induced by the set 
$V(H_X)\subseteq V(S)$.
Then $S_X$ is guarded by the closure of $F_X$.
Since $\varphi$ is a $\Sigma_1$-$L$-sentence, we have}

\claimno=0

{
\newclaim{clm1}{\rt{%
$S\models\varphi$ if and only if $S_X\models\varphi$ for some $X\in\cal X$.}}

}

\rt{%
Thus $S$ will be represented by the collection $\{S_X\}_{X\in\cal X}$ 
of induced substructures.
Updates will be done using Lemma~\ref{lm-databe}, and testing
whether $S\models\varphi$ will be done using (\refclaim{clm1}) and
Lemma~\ref{lm-databe}.}
\end{proof}


The following is a variation of the above theorem for nowhere dense graphs.

\begin{theorem}
\label{thm-databe-nowhere}
Let $L$ be a language with no function symbols, $d_0$ a fixed integer, 
$\varepsilon$ a positive real number and $\GG$ a class of
nowhere-dense graphs.
There exists a data structure representing an $L$-structure $S$ 
\rt{guarded by a member of $\GG$} such that
\begin{itemize}
\item given an \rt{$n$-vertex} graph $G\in\GG$ \rt{and an $L$-structure $S$ guarded by $G$}, 
the data structure is initialized in time $O(n^{1+\varepsilon})$,
\item \rt{if an $L$-structure $S'$ is obtained from $S$}
      by adding or removing a tuple from one of the relations, then 
      the data structure representing $S$ can be changed
      to the one representing $S'$ in time $O(n^{\varepsilon})$ 
      provided that both $S$ and $S'$ are guarded by $G$, and
\item the data structure \rt{allows testing} in time bounded by 
      $O(|\varphi|)$
      whether a given $\Sigma_1$-$L$-sentence $\varphi$ with at most $d_0$ variables is satisfied by $S$,
      and if so, \rt{outputting} one of the satisfying assignments.
\end{itemize}
\end{theorem}

\begin{proof}
Let $\varepsilon>0$.
\rt{The proof makes use of  the data structure from  Theorem~\ref{thm-databe},
which we will refer to as the {\em old data structure}.
The parameters of the latter are now slightly different.
The class $\GG'$ is nowhere dense by Theorem~\ref{thm:augmexp}.
Thus $K$ is no longer a constant;
instead, we may select $K$ to satisfy $K=O(n^{\epsilon/d_0})$,
where $n=|V(G)|$.
The computation of $G'$ takes time $O(n^{1+\varepsilon})$ by 
Lemma~\ref{alg:ndaugment}.
The computation of $c$ takes time $O(|V(G')|+|E(G')|)=O(n^{1+\varepsilon})$,
because $\GG'$ is nowhere dense.
Since $|{\cal X}|=O(n^{\varepsilon})$, the old data structure allows updates
in time $O(n^{\varepsilon})$ and testing $S\models\varphi$ for
$\Sigma_1$-$L$-sentences $\varphi$ in time
$O(|\varphi|n^{\varepsilon})$.
During initialization and after every update we use the old data structure
to compute or recompute the set $\cal S$ of all isomorphism classes
of $L$-structures $A$ with $|V(A)|\le d_0$ such that $A$
is isomorphic to an induced substructure of $S$.
This can be done in time $O(n^{\varepsilon})$, because the size of
$\cal S$ is bounded.
Now $S\models\varphi$ if and only if $A\models\varphi$ for some $A\in\cal S$.
The set $\cal S$ will form the {\em new data structure}, which can be
used to answer queries of the form $S\models\varphi$ in time $O(|\varphi|)$.
}
\end{proof}


\section{Dynamic data structure for first order properties}
\label{sec-dynamic}

In this section, we present our dynamic data structure for testing FO properties.
The main result of this section reads as follows:

\begin{theorem}
\label{thm-dynamic}
Let $\GG$ be a class of graphs with bounded expansion, $L$ a language and
$\varphi$ an $L$-sentence. There exists a data structure that,
\rt{given an $n$-vertex graph $G\in\GG$ and an $L$-structure $A$ guarded 
by $G$, is initialized} in time $O(n)$ and
supports the following operations:
\begin{itemize}
\item adding a tuple to a relation of $A$ in constant time provided $A$ stays guarded by $G$,
\item removing a tuple from a relation of $A$ in constant time, and
\item \rt{answering} in constant time whether $A\models\varphi$.
\end{itemize}
\end{theorem}

Note that in Theorem~\ref{thm-dynamic}, we do not allow to change function values of functions from $L$
to simplify our exposition; this does not present a loss of generality as one can model functions as binary relations.
We first establish a dynamized version of Lemma~\ref{lm-phi-T}.

\begin{lemma}
\label{lm-phi-T-dynamic}
Let $d\ge 0$ be an integer, $L$ a language,
$\varphi(x_0,\ldots,x_n)$ a simple quantifier-free $L$-formula that
is a conjunction of atomic formulas and their negations, and
$T$ a $\varphi$-template.
There exist a language $\overline{L}$ that extends $L$ and
a \rt{(not necessarily simple)}
quantifier-free $\overline{L}$-formula $\overline{\varphi}_T(x_1,\ldots,x_n)$ 
such that the following holds:
\begin{itemize}
\item $\overline{L}$ is obtained from $L$ by adding a function symbol $p$
      and finitely many relation symbols $U_0,\ldots,U_K$ of arity 
      \zz{at most} one,
\item for every rooted forest $F$ of depth at most $d$ and every $L$-structure $S$
      guarded by the closure of $F$,
      there exists an $\overline{L}$-structure $\overline{S}$
      such that $\overline{S}$ is an expansion of $S$ and
       for every $v_1,\ldots,v_n\in V(S)$,
      \begin{center}
      $\overline{S}\models\overline{\varphi}_T(v_1,\ldots,v_n)$
      if and only if
      $S\models\varphi(v_0,v_1,\ldots,v_n)$ for some $v_0\in V(S)$ such that
      \rt{$v_0,v_1,\ldots,v_n$ are compatible with $T,F,S$},
      \end{center}
      where $p^{\overline{S}}$ is the $F$-parent function and
      the relations $U^{\overline{S}}_0,\ldots,U^{\overline{S}}_K$
      can be computed (by listing the singletons they contain) in 
      linear time given $F$ and $S$, and
\item adding or removing a tuple to \rt{or from} a relation of $S$ results
      in adding and\rt{/or} removing a constant number of singletons \rt{to or} from unary relations
      among $U^{\overline{S}}_0,\ldots,U^{\overline{S}}_k$, and
      the changes to all relations $U^{\overline{S}}_0,\ldots,U^{\overline{S}}_k$
      can be computed in constant time, provided $S$ stays guarded by the closure of $F$.

\end{itemize}
\end{lemma}

\begin{proof}
\rt{We construct $\overline L, \overline\varphi_T$ and $\overline S$
as in the proof of Lemma~\ref{lm-phi-T}.}
We need to describe how the relations $U^{\overline{S}}_0,\ldots,U^{\overline{S}}_k$
can be updated in constant time after adding/removing a tuple to/from a relation of $S$.
Let us consider in more detail the main case analyzed in the proof of Lemma~\ref{lm-phi-T};
we leave to the reader the case mentioned at the end of the proof of Lemma~\ref{lm-phi-T} as
the arguments are completely analogous.
Recall (see the proof of Lemma~\ref{lm-phi-T} for notation) that
$U_0(w)$ is the unary relation containing elements $w$ of $F$ at depth $d_v+1$ such that
the subtree of $w$ in $F$ contains an element $v_0$ at depth $d_{x_0}$ (in $F$) with the following properties:
\begin{itemize}
\item there is a $(v_0)$-admissible embedding of the template $T_0$ in $F$ for $S$, and
\item all clauses appearing in the conjunction $\varphi'$ with at least one term from $X_0$ are true with $x_0=v_0$ and
      the terms $t\in X'_0$, say $\alpha_T(t)=q^k(\alpha_T(x_0))$, 
      replaced with $\rt{(}p^{\overline{S}}\rt{)^k}(v_0)$.
\end{itemize}
Since none of the functions of $S$ changes, the first condition cannot change
when adding/removing a tuple to/from a relation of $S$.
The second condition can change only when a tuple containing a term from $X_0$
is added/removed to/from a relation.
But this can result only in a single element (the one at depth $d_v+1$ on the path
in $T$ containing all the elements of the altered tuple) to be added to or removed from $U_0$.
Based on the tuple we add or remove, we can identify this vertex.
The existence of $v_0$ (which must be at depth determined by the template $T_0$)
is tested using the data structure introduced in the proof of Lemma~\ref{lm-databe}:
the values of all terms from $X_0$ with $x_0=v_0$ are in the subtree of $w$ and
those in $X'_0$ are on the path from $w$ to the root. The existence of $v_0$ is equivalent
to the existence of an induced subtree comprised of the path from the root to $w$ and
a subtree of $w$ witnessing that the clauses listed in the second condition are satisfied.
The data structure introduced in the proof of Lemma~\ref{lm-databe} allows testing the existence
of one of these ``witnessing'' subtrees in constant time (assuming the formula $\varphi$ is fixed).
So, we can update the relation $U_0$ in constant time.

Once the relation $U_0$ is updated, the relations $U_1,\ldots,U_k$
can be updated in constant time as well:
we keep a counter at every vertex at depth $d_v$
determining the number of children of that vertex in $U_0$.
\end{proof}

\rt{Next we prove} a dynamized version of Theorem~\ref{thm-bounded+}
 (we state the theorem in the variant with no free variables for simplicity).

\begin{theorem}
\label{thm-dynamic+}
Let $\GG$ be a class of graphs with bounded expansion, $L$ a language,
$\varphi$ an $L$-sentence, \rt{and let $\overline L$, $\overline\GG$,
$\overline\varphi$ and $\cal A$ be as in Theorem~\ref{thm-bounded+}.
Let $A$ and $B$ be $L$-structures guarded by a graph $G\in\GG$, 
let $B$ be obtained from $A$ by adding or deleting a tuple $\tau$ from 
the relation $R^A$ of $A$, and
let $\overline A,\overline G$ and $\overline B,\overline G$ 
be the output of the algorithm $\cal A$
when given $A,G$ and $B,G$, respectively, as input.
Then $\overline B$ can be computed from the knowledge of $\overline A$ and
$\tau$ in constant time.}
\end{theorem}

\begin{proof}
\rt{The proof follows from the proof of Theorem~\ref{thm-bounded+},
using Lemma~\ref{lm-elim} with the proviso that in the proof of
Lemma~\ref{lm-elim}  we use Lemma~\ref{lm-phi-T-dynamic} instead of
Lemma~\ref{lm-phi-T}.}
An important fact is that $A$ and $B$ have the same interpretations of
functions.
We observe that every change in $S$ results
in a constant number of changes in $\overline{S}$ and these changes 
can be identified in constant time.
Hence, in the inductive proof of Theorem~\ref{thm-bounded+}, a single change in $A$
results in constantly many changes to the structure obtained in the first inductive step,
which result in constantly many changes to the structure obtained in the second inductive step (each change
in the structure obtained in the first inductive step yields only constantly 
many changes), and so on.
Since the time to update the final $\overline L$-structure $\overline A$ 
is constant for each of constantly many
choices that propagate through the induction from a single change of $A$, 
the overall update time is constant.
\end{proof} 

\rt{Theorem~\ref{thm-dynamic} follows immediately from Theorem~\ref{thm-dynamic+}.}


\begin{thebibliography}{99}

\bibitem{courcelle}
B.~Courcelle:
{\em The monadic second-order logic of graph 
I.~Recognizable sets of finite graphs},
Inform.~and Comput.~85 (1990), 12--75.

\bibitem{dawar}
A. Dawar, M. Grohe, S. Kreutzer:
{\em Locally excluding a minor},
in: Proc. LICS'07, IEEE Computer Society Press, 270--279.

\bibitem{dak}
A. Dawar, S. Kreutzer:
{\em Parameterized Complexity of First-Order Logic},
Electronic Colloquium on Computational Complexity, TR09-131 (2009).

\bibitem{bib-w1}
R.~G.~Downey, M.~R.~Fellows:
Fixed-parameter tractability and completeness II: On completeness of W[1],
Theoret.~Comput.~Sci.~141 (1995), 109--131.

\bibitem{bib-downey+fellows}
R.~G.~Downey, M.~R.~Fellows:
Parameterized complexity,
Springer, 1999.


\bibitem{wgsurvey}
Z.~Dvo\v{r}\'ak, D.~Kr\'al':
{\em Algorithms for classes of graphs with bounded expansion},
in: Proc. WG'09, LNCS vol.~5911, Springer, 2009, 17--32.

\bibitem{oursoda}
Z.~Dvo\v{r}\'ak, D.~Kr\'al', R.~Thomas:
{\em Coloring triangle-free graphs on surfaces},
in: Proc. SODA'09, ACM\&SIAM, 2009, 120--129.

\bibitem{DvoKraThoFOCS}
Z.~Dvo\v{r}\'ak, D.~Kr\'al', R.~Thomas:
{\em Deciding first-order properties for sparse graphs},
in: Proc. FOCS'10, IEEE, 2010, 133--142.

\bibitem{ourcol}
Z.~Dvo\v{r}\'ak, D.~Kr\'al', R.~Thomas:
{\em Three-coloring triangle-free graphs on surfaces VI. A linear-time algorithm},
in preparation.

\bibitem{bib-eppstein95}
D.~Eppstein:
{\em Subgraph isomorphism in planar graphs and related problems},
in: Proc.~SODA'95, ACM\&SIAM, 632--640.

\bibitem{bib-eppstein99}
D.~Eppstein:
{\em Subgraph isomorphism in planar graphs and related problems},
J.~Graph Algorithms Appl.~3 (1999), 1--27.

\bibitem{bib-eppstein00}
D.~Eppstein:
{\em Diameter and treewidth in minor-closed graph families},
Algorithmica 27 (2000), 275--291.

\bibitem{bib-flum+grohe}
J.~Flum, M.~Grohe:
Parameterized complexity theory,
Birkh\"auser, 2006.

\bibitem{fg}
M. Frick, M. Grohe:
{\em Deciding first-order properties of locally tree-decomposable structures},
J. ACM 48 (2001), 1184--1206.

\bibitem{gaif}
H. Gaifman:
{\em On local and non-local properties},
in: Proc. Herbrands Symp. Logic Coloq., North-Holland, 1982.

\bibitem{gjs}
M. Garey, D. Johnson, L. Stockmeyer:
{\em Some simplified NP-complete graph problems},
Theoret. Comput. Sci. 1 (1976) 237--267. 

\bibitem{grokremeth} M.~Grohe and S.~Kreutzer,
Methods for Algorithmic Meta Theorems,
In Martin Grohe, Johann Makowsky (Eds),
Model Theoretic Methods in Finite Combinatorics,
AMS Contemporary Mathematics Series 558, 
American Mathematical Society, 2011.

\bibitem{kreutzer-survey}
S.~Kreutzer:
{\em Algorithmic meta-theorems},
to appear in a workshop volume for a workshop held in Durham 2006
as part of the Newton institute special programme on Logic and Algorithms.
An extended abstract appeared in:
Proc. IWPEC'08, LNCS vol.~5018, Springer, 2008, 10--12.

\bibitem{bib-nes0}
J. Ne{\v s}et{\v r}il, P. Ossona de Mendez:
{\em Linear time low tree-width partitions and algorithmic consequences},
in: Proc. STOC'06, 391--400.

\bibitem{bib-nes1}
J. Ne{\v s}et{\v r}il, P. Ossona de Mendez:
{\em Grad and classes with bounded expansion I. Decompositions.},
Eur. J. Comb. 29 (2008), 760--776.

\bibitem{bib-nes2}
J. Ne{\v s}et{\v r}il, P. Ossona de Mendez:
{\em Grad and classes with bounded expansion II. Algorithmic aspects.},
Eur. J. Comb. 29 (2008), 777--791.

\bibitem{bib-nes3}
J. Ne{\v s}et{\v r}il, P. Ossona de Mendez:
{\em Grad and classes with bounded expansion III. Restricted graph homomorphism dualities.},
Eur. J. Comb. 29 (2008), 1012--1024.

\bibitem{bib-nowhere}
J.~Ne{\v s}et{\v r}il, P.~Ossona de Mendez:
{\em On nowhere dense graphs},
Eur. J. Comb. 32 (2011), 600--617.

\bibitem{osmenwood}
J. Ne{\v s}et{\v r}il, P. Ossona de Mendez and D.~Wood:
{\em Characterisations and Examples of Graph Classes with Bounded Expansion},
preprint (arXiv:0902.3265v2).

\bibitem{nessurvey}
J. Ne{\v s}et{\v r}il, P. Ossona de Mendez:
{\em Structural properties of sparse graphs},
in: M. Gr{\"o}tschel, G.~O.~H.~Katona (eds.):
Building Bridges Between Mathematics and Computer Science,
Bolyai Society Mathematical Studies vol.~19, Springer, 2008.

\bibitem{bib-niedermeier}
R.~Niedermeier:
Invitation to fixed-parameter algorithms,
Oxford University Press, 2006.


\bibitem{GraphMinors13}
N.~Roberson, P.~D.~Seymour:
{\em Graph minors. XIII: the disjoint paths problem},
J. Combin. Theory Ser. B 63 (1995), 65--110.

\bibitem{Seelinear} D.~Seese:
Linear time computable problems and first-order descriptions,
{\it Mathematical Structures in Computer Science}, {\bf 5} (1996), 505--526.

\bibitem{woodcl}
D.~Wood:
{\em On the maximum number of cliques in a graph},
Graphs Combin. 23 (2007), 337--352.
\end{thebibliography}
\end{document}